\documentclass[runningheads,a4paper]{llncs}
\usepackage{tikz}
\usetikzlibrary{arrows,decorations,backgrounds}
\usepackage[square,sort,comma,numbers]{natbib}
\usepackage{amssymb,amsmath}
\usepackage[mathscr]{euscript}
\usepackage{graphics}
\usepackage{graphicx}
\usepackage[tight]{subfigure}
\usepackage{sgame}
\usepackage{mathrsfs} 
\usepackage{amsfonts}
\usepackage{dsfont}
\usepackage{framed}
\usepackage{fp}
\usepackage{enumerate}
\usepackage{dsfont}
\usepackage{algorithm}
\usepackage{algpseudocode}
\usepackage{appendix}
\algrenewcommand\algorithmicfor{\textbf{Foreach}}
\algrenewcommand\algorithmicwhile{\textbf{While}}
\algrenewcommand\algorithmicdo{\textbf{Do}}
\usepackage{url}

{\makeatletter
 \gdef\xxxmark{%
   \expandafter\ifx\csname @mpargs\endcsname\relax 
     \expandafter\ifx\csname @captype\endcsname\relax 
       \marginpar{xxx}
     \else
       xxx 
     \fi
   \else
     xxx 
   \fi}
 \gdef\xxx{\@ifnextchar[\xxx@lab\xxx@nolab}
 \long\gdef\xxx@lab[#1]#2{{\bf [\xxxmark #2 ---{\sc #1}]}}
 \long\gdef\xxx@nolab#1{{\bf [\xxxmark #1]}}
}

\usepackage[numbers]{natbib}
\newcommand{\RNum}[1]{\uppercase\expandafter{\romannumeral #1\relax}}

\newcommand{\proofsketch}{\vspace*{-1ex} \noindent {\textit{ Proof Sketch.} }}

\begin{document}
\title{Simple and Near-Optimal Mechanisms For Market Intermediation}
\titlerunning{Simple and Near-Optimal Mechanisms For Market Intermediation}
\author{Rad Niazadeh$^{\dag}$, Yang Yuan$^{\dag}$ \and Robert Kleinberg$^{\dag}$
\thanks{All three authors were supported by NSF grant AF-0910940. Robert
Kleinberg was also supported by a Microsoft Research New Faculty Fellowship 
and a Google Research Grant.}}
\authorrunning{Simple and Near-Optimal Mechanisms For Market Intermediation}
\institute{$^{\dag}$Cornell University, Department of Computer Science.}
\toctitle{Simple and Near-Optimal Mechanisms For Market Intermediation}
\maketitle
\begin{abstract}
A prevalent market structure in the Internet economy consists 
of buyers and sellers connected by a platform (such as Amazon
or eBay) that acts as an intermediary and keeps a share of the
revenue of each transaction. 
While the optimal mechanism that maximizes the intermediary's profit
in such a setting may be quite complicated, the mechanisms 
observed in reality are
generally much simpler, e.g.,
applying an affine function to the 
price of the transaction as the intermediary's fee.
\citet{LN07,LN13} initiated the study of such fee-setting 
mechanisms in two-sided markets, and we continue this 
investigation by addressing the question of when an 
affine fee schedule is approximately optimal for worst-case seller distribution. On one hand
our work supplies non-trivial sufficient conditions on the buyer side (i.e. linearity of marginal revenue function, or MHR property of value and value minus cost distributions)
under which an affine fee schedule can obtain a constant
fraction of the intermediary's optimal profit for all seller distributions. On the other hand we complement our result by showing that proper affine fee-setting mechanisms (e.g. those used in eBay and Amazon selling plans) are \emph{unable} to extract a constant fraction of optimal profit in the worst-case seller distribution. As subsidiary results we also show there exists a constant gap between maximum surplus and maximum revenue under the aforementioned conditions. Most of the mechanisms that we propose are also prior-independent with respect to the seller, which signifies the practical implications of our result.
\end{abstract}
\section{Introduction}
\label{sec:intro}

A prevalent market structure in the Internet economy consists 
of buyers and sellers connected by a platform (such as Amazon
or eBay) that acts as an intermediary and keeps a share of the
revenue each time a buyer makes a purchase from a seller. 
What mechanism should the intermediary use to maximize its
profit? In cases the optimal mechanism is unacceptably 
complicated, can simpler mechanisms closely approximate 
the profit of the optimal mechanism?
We approach these questions using the framework of 
Bayesian mechanism design and worst-case approximation guarantees.

To motivate our investigation it is instructive to consider
the transaction fees that are commonly used by intermediaries in
reality. For example, when an item is sold on eBay using a
fixed price listing (as opposed to an auction), the seller is
charged a fee of $0.3+0.1P$, where $P$ is the total amount of 
the sale in 
dollars\footnote{See {\tt http://pages.ebay.com/help/sell/fees.html}}. 
Amazon uses a similar pricing rule for individual sellers, which is 
$0.99+\alpha P$, where $\alpha$ is a real number
determined by the category of the product, typically ranging
from $8\%$ to $15\%$ \footnote{See {\tt http://services.amazon.com/selling/pricing.htm}}.
Generalizing these examples, we say that 
a \emph{fee-setting mechanism} is one in which the intermediary
names a function $w(\cdot)$, the seller names a price $P$, and
the buyer chooses whether or not to take the item at price $P$.
If the transaction takes place, then the intermediary keeps
$w(P)$ and pays $P-w(P)$ to the seller. Otherwise, no money
changes hands. We refer to $w$ as the \emph{fee schedule} of
the mechanism. We say that $w$ is \emph{affine}
if it can be represented in the form $w = (1-\alpha) P+\beta$ for some constants
$\alpha,\beta$,  and we say that an affine schedule $w(P)=(1-\alpha) P+\beta$ is \emph{proper} if  $\alpha\in [0,1], \beta\geq 0$. Note that the fee schedule used by eBay and Amazon (and many other 
intermediaries, for example real estate brokers) 
are affine and proper. 

\citet{LN07,LN13} initiated the study of fee-setting mechanisms in
two-sided markets. They showed that if it is possible for the intermediary
to choose a mechanism that implements a given allocation rule in 
Bayes-Nash equilibrium, then there is a fee-setting mechanism
that does so. They also provided necessary and sufficient conditions
for the intermediary's optimal mechanism to be implemented by an
affine fee-setting mechanism. The necessary and sufficient condition
discovered by \citet{LN07,LN13} requires the seller's cost to be 
drawn from a generalized Pareto distribution (see Definition~\ref{gpd}
below). Using results from extreme value theory, they show that in the 
limit as only the sellers with lowest cost and the buyers with highest
value enter the market, the conditional distribution of the seller's 
cost (conditional on entering the market) approaches a generalized
Pareto distribution, thus providing a partial justification for the
prevalence of affine fee-setting mechanisms in two-sided markets.

Our work draws inspiration from the aforementioned work 
of \citet{LN07,LN13} and seeks a different type of justification
for affine fee-setting mechanisms by asking the question, ``When
are affine fee-setting mechanisms approximately optimal?'' 
Our results pertain to the case 
when the buyer's virtual
valuation function is affine, which 
is the characterization of
generalized Pareto distributions, in 
ex-post IR setting. 
We first show that a specific choice of 
seller prior-independent affine
fee schedule $w(P)=P-\phi_{\mathcal{B}}(P)$ is ex-post IR for every possible seller's
distribution, where prior-independent means the fee schedule only
depends on the buyer's value distribution but not on the seller's
cost distribution. 
Moreover, this affine fee schedule also
achieves a constant-approximation to the 
maximum  surplus --- and hence, also, a constant-approximation
to the optimal revenue. The approximation factor depends on the
exponent of the buyer's generalized power distribution but it is no more than $4$ comparing to optimal intermediary's profit when the buyer's PDF is monotone.
Our results complement the results of \citet{LN07, LN13} in the sense that 
combined with their results, we show that
if either of the buyer side or the seller side 
has affine virtual valuation function, and 
the other side follows regular distributions,
then the best affine fee schedule 
guarantees either optimal or near optimal revenue,
which provides explanation for the phenomenon that
affine fee schedule is widely used in the daily life.

Our second main result explores the setting that the difference between the values of the seller and the buyer follows MHR distribution, which indicates that the surplus and revenue are constant approximation to each other. Under this assumption, we may further extend the buyer's distribution to MHR distributions, and still get constant approximation ratio with constant (and hence affine) fee schedule.

Intriguingly, without proper MHR assumptions the ex-post IR affine fee schedule in the aforementioned
approximation result is \emph{not} proper; 
in contrast to intermediaries in typical
two-sided markets in practice, the intermediary in our
approximation result may charge a transaction fee which is a 
decreasing function of the seller's price. Our third
main result shows that this reliance on improper affine
fee schedules is unavoidable: even when the buyer's value
is assumed to be uniformly distributed on $[0,1]$, there
exist seller cost distributions for which no proper affine
fee-setting mechanism can achieve a constant-approximation
to the optimal revenue.

In the special case that
the buyer's distribution is 
uniform $[0,1]$, we propose an improved mechanism, 
which gives $3$-approximation fee-setting mechanism
to the optimal revenue.
We also prove that if one needs a prior independent 
affine fee schedule when the buyer's distribution 
is uniform $[0,1]$, then 
$\alpha-\beta=1$ is 
necessary. Moreover, among all the prior independent 
affine fee schedule, 
$w(P)=1-P$ gets the best approximation ratio $8$ comparing to maximum surplus.
From this perspective, our proposed affine fee schedule is 
optimal.
Finally, our proof techniques reveal the fact that there exists a constant gap between optimal revenue and maximum surplus  when buyer's distribution is generalized Pareto distribution as a side dish.

The primary source of difficulty in proving these results is that
fee-setting mechanisms are not Bayes-Nash incentive compatible (BNIC). 
Thus, deriving a revenue guarantee for the intermediary requires
first solving for the Bayes-Nash equilibrium of the mechanism.
Our paper adopts the approach introduced 
by \citet{LN07,LN13} for deriving the Bayes-Nash
equilibrium. The technical heart of our paper lies in
some surprising connections between the affine fee schedule, Bayes-Nash equilibrium payment function, and the cumulative hazard rate function.
These connections are non-trivial, which
make the proof succinct while the results
are still general.
Starting from that, we got expressions of 
the three quantities of interest ---
the maximum surplus, the optimal revenue, and the affine fee-setting mechanism's revenue
--- in a closely related form.
Then, leveraging our assumption
that the buyer's virtual value function is affine, we are able
to choose an affine fee schedule 
to approximate the optimal revenue.


\subsection{Related Work}
\citet{MS83} showed that for one seller one buyer setting, 
if there is no intermediary
between them, then no incentive-compatible individually rational 
mechanism can produce post efficient outcome, where post efficient outcome
means the trade should take place whenever the buyer's value is larger than 
the seller's cost. Based on this impossibility result, they also
considered the case that intermediary is allowed, and
both the seller and the buyer can trade with the intermediary only. 

\citet{DGTZ12} 
studied the double auction, in which the 
intermediary designs mechanism for the 
buyers and the sellers to extract maximum revenue. 
In the paper, they provided optimal or near-optimal
mechanisms for  
both single dimensional and multi-dimensional environments with continuous or discrete distributions.
\citet{JW12} studied the same problem 
with single unit-demand buyer and multiple sellers, and gave a characterization for the optimal solution in this setting. Since the optimal mechanism is generally hard to implement, they also proposed several approximation mechanisms, including picking the best item and sell, or using anonymous virtual reserve price combined with greedy algorithm. 

Contract problem has a similar setting as the intermediary problem:
the principle (intermediary) proposed a contract ($w(\cdot)$ function) to the 
agent (the seller), and the agent will choose his action and get 
a output ($P$ payment), and then give the principle $w(P)$, keep
$P-c$ as its utility. 
Previously, researchers have found evidence showing that linear contract is 
powerful in this setting.
\citet{PBS07} studied linear contract problem, and found that 
linear contracts are common in practice not only because the simplicity, 
but also due to the fact that the optimal linear contract guarantees
at least $90\%$ of the fully optimal contract in the canonical moral hazard setting.
\citet{GC13} proved that under mild assumptions, 
the optimal contract is actually linear.

Simple mechanisms and their approximation ratios to the corresponding optimal
mechanisms have been an important research topic in the literature. For example, 
\citet{BK94} showed that in the i.i.d., regular, 
single dimensional setting, second price auction 
with $n+1$ bidder will give more revenue
than the optimal auction with $n$ bidders.
\citet{HR10} investigated the single dimensional 
setting where bidders have independent valuations, 
and showed that VCG with anonymous reserve price
can achieve $4$-approximation to the optimal 
revenue. \citet{DRY10} considered the auctions that 
are prior-independent, in the sense that
the auction will achieve good approximation 
to the optimal revenue while the specific value distributions
of the bidders
are not used in the auction.

\section{Preliminaries}
In this paper we consider the problem of single-item trade, in which  a profit-maximizing broker mediates the exchange between a buyer and a seller. In particular, we follow the Bayesian mechanism design approach wherein a Bayesian designer looks to find the trade mechanism with the maximum possible revenue in expectation over the distributions from which the preferences of the buyer and seller are drawn. We assume the preferences of buyer and seller are private values drawn from  product distributions, which are common knowledge. 
\subsection{Setting, notations, solution concepts, and basics}
We assume the reader is familiar with the general model of single dimensional mechanism design for risk neutral agents, including the definitions of incentive compatibility and individual rationality, basics of Bayesian mechanism design, and adapting these concepts to the exchange setting (see Appendix~\ref{mechbasic}). Still, it is worth identifying a few aspects of our notations and terminology.

Suppose the seller $\mathcal{S}$ has a private cost $c$ 
and the buyer $\mathcal{B}$ has a private value $v$ for the item. We use $F$ (and $f$) to denote 
the CDF(and PDF) of $v$, 
and $G$ (and $g$) to denote
the CDF (and PDF) of $c$. Unless stated otherwise, we assume the support of $f$ is $[0,\overline{v}]$ and the support of $g$ is $[0,\overline{c}]$. We define the marginal revenue functions (a.k.a. \textit{virtual preferences}) of seller and buyer as follows. Let $\phi_{\mathcal{S}}(c)\triangleq c+\frac{G(c)}{g(c)}$   be defined as the virtual cost of the seller and $\phi_{\mathcal{B}}(v)\triangleq v-\frac{1-F(v)}{f(v)}$ be defined as the virtual value of the buyer. We also define buyer's \textit{hazard rate}, $h_B(v)\triangleq\frac{f(v)}{1-F(v)}$, and \textit{cumulative hazard rate}, $H_B(v)\triangleq\int_{0}^{v}h_B(z)dz$. It can be easily shown that $1-F(v)=e^{-H_B(v)}$, which is a famous property of cumulative hazard rate. 

We say a buyer (or a seller) is \textit{buyer-regular} (or \textit{seller-regular}) if $\phi_{\mathcal{B}}(v)$ (or $\phi_{\mathcal{S}}(c))$ is monotone non-decreasing. A buyer's distribution is said to be \textit{MHR (monotone hazard rate)} if $h_\mathcal{B}(v)$ is monotone non-decreasing (or equivalently $H_\mathcal{B}(v)$ is convex). For a regular buyer $v$, \textit{monopoly price} is defined to be $\eta_v=\phi_{\mathcal{B}}^{-1}(0)$ (i.e. if $v\geq\eta_v$ virtual value is non-negative). Moreover, \textit{monopoly revenue} $R_{\eta}^{v}\triangleq\eta_v(1-F(\eta_v))$ is the expected revenue one gets by posting $\eta_v$ to a buyer with value $v$.
\subsection{Characterization of distributions with affine virtual value/cost}
\label{paretosec}
A critical constraint throughout this paper, which is appearing in different forms in many of our results and background results on this subject, is when the buyer or seller has an affine virtual preference, i.e. when $\phi_\mathcal{S}(c)=xc+y$ or when $\phi_\mathcal{B}(v)=xv-y$ for $x,y\in\mathbb{R}$.  We now characterize the buyer distributions and seller distributions with the above property as follows.
\begin{definition}
\label{gpd}
A generalized Pareto distribution $F$ with parameters $\mu,\lambda,$ and $\xi$, where $\mu,\lambda,\xi\in\mathbb{R}$, $\lambda > 0$ and $\xi \geq  0$, is defined by the following cumulative density function. 
\begin{equation}
F(x)= \left\{ \begin{array}{rl}
1-(1-\xi \lambda (x-\mu))^{\frac{1}{\xi}} &\mbox{ if $\xi>0$} \\
1-e^{\lambda(x-\mu)} &\mbox{~if $\xi=0$}
\end{array} \right.\nonumber
\end{equation}
and the support is bounded and equal to $[\mu,\mu+\frac{1}{\xi\lambda}]$ if $\xi >0$, and is unbounded and equal to $[\mu,+\infty)$ if $\xi=0$. When $\xi>0$ we refer to the distribution as \emph{generalized power distribution} and when $\xi=0$ we refer to it as \emph{generalized exponential distribution}.
\end{definition}
It is worth mentioning that the family of Pareto distributions are skewed, heavy-tailed distributions that are sometimes used to model the distributions of incomes and other financial variables. For the cost of the seller, we define a similar distribution as follows.
\begin{definition}
\label{reverspareto}
The seller with cost $c$ has a reverse-generalized Pareto distribution with parameters $\mu, \lambda,$ and $\xi$ if $-c$ is a random variable drawn from a generalized Pareto distribution with parameters $\mu,\lambda$ and $\xi$.
\end{definition}

For generalized Pareto distribution family, one can easily prove the following corollary by definition.
\begin{corollary}
\label{cor:gpd}
If $v$ is drawn from a generalized Pareto distribution with parameters $\mu, \lambda,$ and $\xi$, then $\phi_B(v)=(1+\xi)v-(\frac{1}{\lambda}-\xi\mu)$. If $c$ is drawn from a reverse-generalized Pareto distribution with parameters $\mu, \lambda,$ and $\xi$, then $\phi_S(c)=(1+\xi)c+(\frac{1}{\lambda}+\xi\mu)$.
\end{corollary}
We can also prove that the inverse is true, i.e. affine virtual preferences implies the generalized Pareto distribution (To prove this lemma, simply solve the corresponding
differential equations coming from the definitions, which we omit
here)
\begin{lemma}
\label{lemma:gpd}
A buyer (or seller) has affine virtual value (or cost) only if its value (or cost) is drawn from  a generalized Pareto distribution (or reverse-generalized Pareto distribution).
\end{lemma}
\section{Background results}
\label{backresult}
In this section, we investigate a class of mechanisms known as \emph{fee-setting}, introduced first by \citet{LN07}. In these mechanisms,  the intermediary asks the seller to bid her preferred price. If a buyer is willing to buy the item with this price, the intermediary takes a share of the trade money and gives the rest to the seller. Fee-setting mechanisms are simple, intuitive, easy to implement and more robust compared with Myerson's optimal mechanism. 
\subsection{Fee-setting exchange mechanisms}
\label{indopt}
 We first define a \emph{fee-setting mechanism} as follows.
\begin{definition}
A \emph{a fee-setting mechanism} with common knowledge \textit{fee schedule} $w(.):\mathbb{R}\rightarrow \mathbb{R}$ is an indirect mechanism for single buyer single seller exchange that runs the following steps subsequently:
\begin{itemize}
\item Trader asks the seller to bid its desired price $P$,
\item Trader then posts the price $P$ for the buyer,
\item If $v <P$ then the trade doesn't happen and all the payments will be zero.
\item If $v \geq P$ then the item will be traded, trader charges the buyer $P$, keeps its share of the trade $w(P)$, and pays $P-w(P)$ to the seller.
\end{itemize}
\end{definition}
 We now define \emph{affine fee setting exchange mechanisms} formally below. 
\begin{definition} An \emph{affine  fee setting exchange mechanism} with parameters $\alpha$ and $\beta$ is a fee-setting mechanism with affine fee schedule $w(P)\triangleq(1-\alpha)P+\beta$,  $\alpha, \beta \in \mathbb{R}$. 
\end{definition}
In this paper, we refer to an affine exchange fee mechanism with parameters $\alpha$ and $\beta$ as $\textrm{APX}(\alpha,\beta)$. We also define $\textrm{Rev-APX}(\alpha,\beta)$ to be the revenue of $\textrm{APX}(\alpha,\beta)$ when strategy profile of agents is a BNE (As we will discuss below, for affine exchange mechanisms there is a unique BNE  under the regularity assumption). Moreover, $\textrm{OPT-Rev}$ is defined to be the revenue of optimal Myerson mechanism, and $\textrm{OPT-Surplus}$ to be the surplus of VCG mechanism.
\subsection{Characterization of BNE strategy of the seller }
By a standard argument similar to those used in the Bayes-Nash equilibrium characterization of single dimension mechanism~\cite{RB78} one can characterize the BNE of the fee-setting mechanism. More formally, we have the following theorem, proved in~\cite{LN07}, that characterizes the BNE of the fee-setting mechanisms. 
\begin{theorem}\label{bnechar} \cite{LN07} 
Consider a fee-setting mechanism with differentiable fee-setting $w(.)$, then $P:{[0,\overline{c}]}\rightarrow\mathbb{R}^{+}$ is a BNE strategy of the seller if and only if:
\begin{itemize}
\item $P(c)$ is monotone non-decreasing with respect to $c$.
\item $P(c)$ satisfies $\phi_\mathcal{B}(P(c))=P(c)-\frac{P(c)-w(P(c))-c}{1-\frac{\partial w}{\partial p}(P(c))}$.
\end{itemize} 
\end{theorem}

Although the characterization in Theorem~\ref{bnechar} is indirect, it has many nice implications in the special case of fee-settings mechanisms with affine fee schedule.
\begin{corollary}
\label{affine:bnechar}
Suppose in an exchange setting seller is regular. Then for an affine fee-setting mechanism with fee schedule $w(P)=(1-\alpha)P+\beta$,  $P(c)=\phi_\mathcal{B}^{-1}(\frac{c+\beta}{\alpha})$ is the unique BNE strategy of seller.
\end{corollary}
\begin{proof}
From Theorem~\ref{bnechar}  we know in any BNE, we have
\begin{equation}
\phi_\mathcal{B}(P(c))=P(c)-\frac{P(c)-w(P(c))-c}{1-\frac{\partial w}{\partial p}(P(c))}=P(c)-\frac{\alpha P(c)+\beta+c}{1-(1-\alpha)}=\frac{c+\beta}{\alpha}\nonumber
\end{equation}
and as buyer is regular, $\phi_\mathcal{B}$ is invertible, so in any BNE $P(c)=\phi_\mathcal{B}^{-1}(\frac{c+\beta}{\alpha})$.
\end{proof}
\subsection{Optimality of affine and non-affine fee-setting mechanisms}
Considering the class of fee-setting mechanisms, one important question is how well these mechanisms can perform comparing to Myerson's optimal mechanism. \citet{LN07, LN13} showed that with a proper choice of function $w(P)$ (not necessarily affine) one can design a fee-setting mechanism that extracts the same revenue in expectation as in Myerson's optimal mechanism. While this result is surprising by itself, they also could show that optimal fee-setting mechanism will be affine when seller's cost is drawn form a reverse-generalized Pareto distribution as in Definition~\ref{reverspareto} (in other words, when the seller's virtual cost is affine). For more details on this result and a simple proof using revenue equivalence theorem~\cite{RB78}, see Appendix~\ref{appsec3}. 

\section{Main results}
\label{mainresult}
As can be seen from the discussion in the last section, \citet{LN07} initiated the study of affine fee-setting 
mechanisms in two-sided markets and identified
necessary and sufficient conditions for the intermediary's
optimal fee schedule to be affine for worst-case buyer distribution. In this section, we continue this 
investigation by addressing the question of when an 
affine fee schedule is optimal or approximately optimal for worst-case seller distribution. By simulation, one can show that there exists a pair of seller and buyer distributions for which the best affine mechanism is not optimal (for example see~\cite{LN07}). However,
in those cases, we may still be able to get constant approximations to maximum intermediary profit with affine fee-settings. We have three main results following this line of thought.

As our first result, intuitively when at least one side of the bilateral market has some linear behaviors it might be possible for the mechanism designer to extract optimal or approximately optimal revenue from the buyer and seller using affine fee-settings. Under this condition, we propose improper fee-setting mechanisms that can extract constant approximations to optimal revenue. More formally:
\begin{quote}
\textbf{Main Result 1}~\emph{ If the buyer has affine virtual value, under some mild assumptions, the affine fee-setting mechanism $w(P)=P-\phi_\mathcal{B}(P)$  extracts a constant approximation of optimal intermediary's revenue in expectation for any seller-regular distributions. Moreover, optimal intermediary's revenue and maximum surplus are in constant approximation of each other in expectation. }\\
\end{quote}

As the second result, when surplus and revenue are in constant approximation of each other (for example when the distributions involved in the trade are not heavy-tailed) posting a proper price for the buyer can always extract constant approximations to optimal surplus, and hence optimal revenue, and seller's cost will not be an important issue. More formally:
\begin{quote}
\textbf{Main Result 2}~\emph {If the random variables $v$ (buyer's value) and $v-c$ (difference of buyer's value and seller's cost) are MHR, the constant fee-setting mechanism $w(P)=\eta_{v-c}$ extracts  constant approximations to optimal intermediary's revenue in expectation for any seller distributions, in which $\eta_{v-c}$ is the monopoly price for the random variable $v-c$ ($\eta_{v-c}=\phi_{v-c}^{-1}(0)$). Moreover, optimal intermediary's revenue and maximum surplus are in constant approximation of each other in expectation.}
\end{quote}

As the final result, we show that a mechanism designer who tries to get constant approximation to optimal revenue for all seller's distribution (especially for heavy-tailed distributions), cannot avoid using the improper fee-setting mechanisms. Formally:
\begin{quote}
\textbf{Main Result 3}~\emph{Even when the buyer's value is drawn from $\textrm{unif}~[0,1]$, there
exists seller cost distributions for which no proper affine
fee-setting mechanism can achieve a constant-approximation
to the optimal intermediary's revenue.}
\end{quote}

In the next Section, we first provide a proof sketch for our first main result, and then for the special case when buyer's value is uniform we propose an improved fee-setting mechanism accompanied by a refined analysis, which  gives us a better approximation ratio. Then in Section~\ref{sec:MHR} we sketch the proof of second main result. Finally in Section~\ref{sec:inapprox} we elaborate on our third inapproximability result.

\subsection{Approximations for affine buyer's virtual value}
Suppose buyer's virtual value is affine, i.e. $\phi_{\mathcal{B}}(v)=\alpha v -\beta$,\footnote{Note that due to Corollary~\ref{cor:gpd} and Lemma~\ref{lemma:gpd}, $v$ is drawn from a generalized Pareto distribution, and hence $\alpha$ should be in $[1,\infty)$.} and now look at the affine fee-setting mechanism $w(P)=P-\phi_\mathcal{B}=(1-\alpha)P+\beta$. We start by proving the following properties of this mechanism, which also show the mechanism is ex-post IR for seller, buyer and trader (and hence no party regrets attending the trade).
\begin{lemma} 
\label{affine:property}If $\phi_{\mathcal{B}}(v)=\alpha v -\beta$ and $P(c)$ is the BNE strategy of seller, then affine fee-setting mechanism $w(P)=P-\phi_\mathcal{B}(P)=(1-\alpha)P+\beta$ has the following 
properties:
\begin{enumerate}[(a)]
\item $\forall c:$ $w(\phi_\mathcal{B}(P))=\alpha w(P)$ and $\phi_\mathcal{B}(\phi_\mathcal{B}(P))=c$.
\item Ex-post utilities of seller and trader are always non-negative.
\item $\forall v: e^{-H_\mathcal{B}(v)}=(\frac{w(v)}{\beta})^{\frac{1}{\alpha-1}}$, when $\alpha\neq 1$.
\end{enumerate}
\end{lemma}
\begin{proof}
To prove (a) we have $w(\phi_\mathcal{B}(P))=(1-\alpha)(\alpha P-\beta)+\beta=\alpha((1-\alpha)P+\beta)=\alpha w(P)$. Moreover, due to Corollary~\ref{affine:bnechar}, $\phi_\mathcal{B}(P)=\frac{c+\beta}{\alpha}$ and hence $c=\alpha \phi_\mathcal{B}(P) -\beta=\phi_\mathcal{B}(\phi_\mathcal{B}(P))$.
To prove (b), note that utility of trader is equal to $w(P)=P-\phi_\mathcal{B}(P)\geq 0$, due to properties of virtual value. Also, seller's ex-post utility when trade happens is equal to $P-w(P)-c=\phi_\mathcal{B}(P)-c\geq \phi_\mathcal{B}(\phi_\mathcal{B}(P))-c=0$, due to property (a). To prove (c) we have $h_\mathcal{B}(v)=(v-\phi_\mathcal{B}(v))^{-1}=(w(v))^{-1}$. Now, the following calculation finds cumulative hazard rate $H_\mathcal{B}(v)$  which completes the proof of (c).
\begin{equation*}
H_\mathcal{B}(v)=\int_0^{v}h_\mathcal{B}(z)dz=\frac{\ln(w(v))}{1-\alpha}-\frac{\ln(w(0))}{1-\alpha}=\frac{\ln(\frac{w(v)}{\beta})}{1-\alpha}~~~~~~~~~~~~~~~~~~~~~~~~~~~~~~~~~~~~~\qed
\end{equation*}
\end{proof}

Now, using the above properties we prove one can extract a constant portion of optimal revenue and optimal surplus by the above mechanism. The intuition behind the proof is as follows. Look at the special case when the buyer's distribution is uniform on  $[0,1]$.
Then the fee schedule that we propose is $w(P)=1-P$. 
At the first glance this appears counterintuitive: as a seller, if you ask for a higher price then the broker gets less money from you.
But the seller needs to take a trade-off when setting the price:
if the seller picks $P=1$, which minimizes the broker's fee as $w(1)=0$, 
then the chance of finding a buyer with this price will be zero, which 
produces zero utility to the seller. So the seller needs to find 
a balanced price, at which the chance of finding a buyer is large, 
and the fee paid to the intermediary is reasonable as well. 
In other words, the seller is buying ``chance of trade'' from the intermediary by paying $1-P$ to it. We formalize this argument by the following theorem. (Figure~\ref{fig1} presents a geometric proof sketch.)
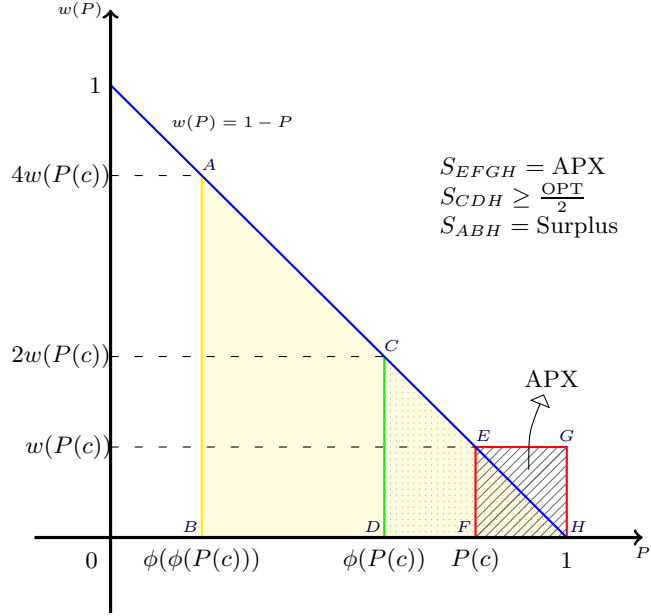
\begin{figure}
\centering
\tikzstyle{letterlabel}=[thick,font=\fontsize{6}{6},color=black!70!blue]
\begin{tikzpicture}
\draw (7,-0.2) node [font=\fontsize{6}{6}] {\emph{$P$}};
\draw (-0.25,-0.3) node {$0$};
\draw (6,-0.3) node {$1$};
\draw (-0.4,7) node[font=\fontsize{6}{6}] {\emph{$w(P)$}};
\draw (-0.2,6) node {$1$};
\draw (1.6,5.5) node[font=\fontsize{6}{6}] {$w(P)=1-P$};
\foreach \x in {1.2}
\foreach \ydel in {-0.65}
\foreach \plabel in {P(c)}
\foreach \xlabel in {0.15}
\foreach \shx in {-0.15}
{
\filldraw[color=yellow!15!white] (6-\x*4,0) -- ++(0,4*\x) -- ++(\x*4,-\x*4) -- ++(-\x*4,0);
\draw (6-\x,-0.3) node {$\plabel$};
\draw[color=red,thick] (6-\x,0)--(6-\x,\x)--(6,\x) -- (6,0);
\draw[loosely dashed] (0,\x)--(6-\x,\x);
\draw (\ydel+0.1,\x) node {$w(\plabel)$};
\draw (6-\x*2,-0.3) node {$\phi(\plabel)$};
\draw[color=green,thick] (6-\x*2,0)--(6-\x*2,\x*2);
\draw[loosely dashed] (0,\x*2)--(6-\x*2,\x*2);
\draw (\ydel,2*\x) node {$2w(\plabel)$};
\draw (6-\x*4,-0.3) node {$\phi(\phi(\plabel))$};
\draw[color=yellow,thick] (6-\x*4,0)--(6-\x*4,\x*4);
\draw[loosely dashed] (0,\x*4)--(6-\x*4,\x*4);
\draw (\ydel,4*\x) node {$4w(\plabel)$};
\draw (6-\x+\shx,\xlabel) node[letterlabel] {$F$};
\draw (6-2*\x+\shx,\xlabel) node[letterlabel] {$D$};
\draw (6-4*\x+\shx,\xlabel) node[letterlabel] {$B$};
\draw (6-\shx,\xlabel) node[letterlabel] {$H$};
\draw (6-4*\x+0.1,4*\x+\xlabel) node[letterlabel] {$A$};
\draw (6-2*\x+0.1,2*\x+\xlabel) node[letterlabel] {$C$};
\draw (6-\x+0.1,\x+\xlabel) node[letterlabel] {$E$};
\draw (6,\x+\xlabel) node[letterlabel] {$G$};
\foreach \totstep in {10}
\foreach \step in {0,...,\totstep}
{
\draw[very thin,color=black!60!white] (6-\x, \step/\totstep*\x) -- (6-\step/\totstep*\x, \x);
\draw[very thin,color=black!60!white] (6, \step/\totstep*\x) -- (6-\step/\totstep*\x, 0);
}
\foreach \totstep in {30}
\foreach \step in {0,...,\totstep}
{
\draw[very thin,dotted, color=black!60!white] (6-2*\x, \step/\totstep*2*\x) -- (6-\step/\totstep*2*\x, \step/\totstep*2*\x);
}
}
\draw[-open triangle 90, thin] (5.5, 0.9).. controls +(0,0.5).. +(0.2,1);
\draw (5.5,4.5) node[font=\small] {$\begin{array}{l}S_{EFGH}=\mathrm{APX}\\S_{CDH}\geq \frac{\mathrm{OPT}}{2}\\S_{ABH}=\mathrm{Surplus}\end{array}$};
\draw (5.8,2.1) node[font=\small] {APX};
\draw[color=blue,thick] (6,0) -- (0,6);
\draw[->,very thick] (-1,0) -- (7,0);
\draw[->,very thick] (0,-1) -- (0,7);
\end{tikzpicture}
\caption{In this figure, buyer value is unif$~[0,1]$ and $w(P)=1-P$. This fee-setting mechanism (APX) extracts $\frac{1}{4}$ fraction of optimal revenue (OPT) and $\frac{1}{8}$ fraction of optimal surplus (Surplus) in expectation, which can be seen by comparing the area of corresponding regions. }
\label{fig1}
\end{figure}
\begin{theorem}
\label{theorem:main1}
Suppose buyer's virtual value is affine: $\phi_\mathcal{B}(v)=\alpha v- \beta$ for some $\alpha\geq 1$. Then the revenue of affine fee-setting mechanism $w(P)=P-\phi_\mathcal{B}(P)$ is $\alpha^{\frac{1}{\alpha-1}}-$approximation to optimal revenue and $\alpha^{\frac{\alpha+1}{\alpha-1}}-$approximation to optimal surplus in expectation. 
\end{theorem}
\begin{proof}
Suppose $P(c)$ is seller's BNE strategy. We first find an equivalent expression for the expectation of revenue of the fee-setting mechanism $w(P)=P-\phi_\mathcal{B}(P)=(1-\alpha)P+\beta$  (i.e. $\textrm{APX}(\alpha,\beta)$) conditioned on a fixed cost $c$ for the seller. For a fixed $c$, trade happens if buyer's value is at least equal to the posted price, i.e. if $v\geq P(c)$. In this case, trader gets its share $w(P)$ and returns the rest to the seller. So,
\begin{equation}
\mathbb{E}\{\textrm{Rev-APX}|c\}=\mathbb{E}\{w(P)\mathds{1}\{v\geq P\}|c\}=w(P)(1-F(P))=w(P)e^{-H_\mathcal{B}(P)}\overset{(i)}{=}\frac{w(P)^{\frac{\alpha}{\alpha-1}}}{\beta^{\frac{1}{\alpha-1}}},\nonumber
\end{equation} 
where equality (i) is due to property (c) in Lemma~\ref{affine:property}. Basically, the conditional revenue of $\textrm{APX}$ is the measured area under a rectangle with width equal to $w(P)$ and length equal to the interval $\{v: v\geq P\}$, when we use the distribution of buyer's value as the measure function. For the special case of uniform distribution, this corresponds to normal area of this rectangle, as can be seen in Figure~\ref{fig1}. For the optimal revenue, we try to obtain a similar upper-bound. From Myerson's theory of optimal mechanisms, we know $\mathbb{E}\{\textrm{OPT-Rev}|c\}=\mathbb{E}\{(\phi_\mathcal{B}(v)-\phi_{\mathcal{S}}(c))\mathds{1}\{\phi_\mathcal{B}(v)\geq \phi_{\mathcal{S}}(c)\}|c\}$. By plugging in buyer's affine virtual function $\phi_\mathcal{B}(v)=\alpha v-\beta$ in this expression, we have:
\begin{align}
&\mathbb{E}\{\textrm{OPT-Rev}|c\}=\mathbb{E}\left\{(\alpha v-\beta-\phi_{\mathcal{S}}(c))\mathds{1}\{\alpha v-\beta\geq \phi_{\mathcal{S}}(c)\}|c\right\}\nonumber\\
&=\alpha \mathbb{E}\left\{(v-\frac{\phi_{\mathcal{S}}(c)+\beta}{\alpha})\mathds{1}\{v\geq \frac{\phi_{\mathcal{S}}(c)+\beta}{\alpha}\}|c\right\}\overset{(i)}{=}\alpha\mathbb{E}\left\{\frac{1}{h_\mathcal{B}(v)}\mathds{1}\{v\geq \frac{\phi_{\mathcal{S}}(c)+\beta}{\alpha}\}|c\right\}\nonumber\\
&\overset{(ii)}{\leq} \alpha\mathbb{E}\left\{\frac{1}{h_\mathcal{B}(v)}\mathds{1}\{v\geq \phi_\mathcal{B}(P)\}|c\right\}\overset{(iii)}{=}\alpha\mathbb{E}\left\{w(v)\mathds{1}\{v\geq \phi_\mathcal{B}(P)\}|c\right\}\nonumber,
\end{align}
where equality (i) is true because for any $x$, $\mathbb{E}\{(v-x)\mathds{1}\{v\geq x\}\}=\int_{t\geq x}(1-F(t)) dt=\mathbb{E}\{\frac{1}{h(v)}\mathds{1}\{v\geq x\}\}$, (ii) is true because $\frac{\phi_{\mathcal{S}}(c)+\beta}{\alpha}\geq \frac{c+\beta}{\alpha}=P(c)$, and (iii) is true because $w(v)=v-\phi_\mathcal{B}(v)=\frac{1}{h_\mathcal{B}(v)}$. The last upper-bound on the conditional revenue of optimal Myerson divided by $\alpha$ is the measured area under the curve $w(v)$ in the interval $\{v: v\geq \phi_\mathcal{B}(P)\}$, when we use the distribution of buyer's value as the measure function. For the special case of uniform distribution, this corresponds to normal area under the curve, as can be seen in Figure~\ref{fig1}. One can calculate this term by taking the integral of an affine function and show this is equal to $\alpha^{\frac{1}{\alpha-1}}\mathbb{E}\{\textrm{Rev-APX}|c\}$, which gives us the desired approximation factor (see Appendix~\ref{detail:proof}). Here we take a different approach which results in a slightly weaker approximation factor, but is more intuitive. $w(v)$ is non-increasing (as $\alpha\geq 1$) and hence $w(v)\leq w(\phi_\mathcal{B}(P))$ in the region $\{v:v\geq \phi_\mathcal{B}(p)\}$. So, we can upper-bound the conditional expectation of optimal revenue further by
\begin{align*}
&\mathbb{E}\{\textrm{OPT-Rev}|c\}\leq \alpha w(\phi_\mathcal{B}(P))(1-F(\phi_\mathcal{B}(P)))=\alpha w(\phi_\mathcal{B}(P))e^{-H_\mathcal{B}(w(\phi_\mathcal{B}(P)))}\nonumber\\
&\overset{(i)}{=}\alpha\frac{w(\phi_\mathcal{B}(P))^{\frac{\alpha}{\alpha-1}}}{\beta^{\frac{1}{\alpha-1}}}\overset{(ii)}{=}\alpha^{\frac{\alpha}{\alpha-1}}\mathbb{E}\{\textrm{Rev-APX}|c\},
\end{align*} 
in which (i) is true due to property (c) in Lemma~\ref{affine:property}, and (ii) is true because $w(\phi_\mathcal{B}(P))=\alpha w(P)$ based on property (a) in Lemma~\ref{affine:property}. Taking expectation with respect to $c$ will prove the desired approximation factor with respect to optimal revenue. 

To compare the revenue of our fee-setting mechanism with the surplus, we use the same machinery to find an expression for the expectation of maximum surplus for a fixed $c$. Similar to the calculations for optimal revenue we have $\mathbb{E}\{\textrm{OPT-Surplus}|c\}=\mathbb{E}\{(v-c)\mathds{1}\{v\geq c\}|c\}=\mathbb{E}\{\frac{1}{h_\mathcal{B}(v)}\mathds{1}\{v\geq c\}|c\}=\mathbb{E}\{w(v)\mathds{1}\{v\geq c\}|c\}=\mathbb{E}\{w(v)\mathds{1}\{v\geq \phi_\mathcal{B}(\phi_\mathcal{B}(P))\}|c\}$, where the last equality is true because $\phi_\mathcal{B}(\phi_\mathcal{B}(P))=c$ due to property (a) in Lemma~\ref{affine:property}. Again, the conditional maximum surplus is the measured area under the curve $w(v)$ in the interval $\{v: v\geq \phi_\mathcal{B}(\phi_\mathcal{B}(P))\}$, when we use the distribution of buyer's value as the measure function. For the special case of uniform distribution, again this corresponds to normal area under the curve, as can be seen in Figure~\ref{fig1}. Now, again one can either calculate this term by taking integral of an affine function and show this is equal to $\alpha^{\frac{\alpha+1}{\alpha-1}}\mathbb{E}\{\textrm{Rev-APX}|c\}$ (which gives us the desired approximation factor, see Appendix~\ref{detail:proof} for the proof), or can use the following upper-bound for a slightly weaker factor (but more intuitive).
\begin{align*}
&\mathbb{E}\{\textrm{OPT-Surplus}|c\}\overset{(i)}{\leq}w(\phi_\mathcal{B}(\phi_\mathcal{B}(P)))(1-F(\phi_\mathcal{B}(\phi_\mathcal{B}(P))))= w(\phi_\mathcal{B}(\phi_\mathcal{B}(P)))e^{-H_\mathcal{B}(w(\phi_\mathcal{B}(\phi_\mathcal{B}(P))))}\nonumber\\
&=\frac{w(\phi_\mathcal{B}(\phi_\mathcal{B}(P)))^{\frac{\alpha}{\alpha-1}}}{\beta^{\frac{1}{\alpha-1}}}\overset{(ii)}{=}\alpha^{\frac{2\alpha}{\alpha-1}}\mathbb{E}\{\textrm{Rev-APX}|c\},
\end{align*} 
where (i) is true because $w(v)$ is non-increasing, and (ii) is true because by using property (a) of Lemma~\ref{affine:property} twice we have $w(\phi_\mathcal{B}(\phi_\mathcal{B}(P)))=\alpha w(\phi_\mathcal{B}(P))=\alpha^2w(P)$. Taking expectation with respect to $c$ will complete the proof. \qed
\end{proof}

We are now ready to obtain approximation ratios for different cases of generalized Pareto distributions, namely general power distributions and exponential distributions. This is exactly the same class of distributions that \citet{LN07,LN13} investigated.

\begin{corollary}[Exponential distribution]
Suppose $F(v)=1-e^{-\lambda v}$ over $[0,\infty)$ for $\lambda>0$. Then revenue of $\textrm{APX}(1,\frac{1}{\lambda})$ (i.e. fee-setting with $w(P)=\frac{1}{\lambda}$) is $e^2$-approximation to maximum surplus, and $e$-approximation to the optimal revenue in expectation.
\end{corollary}
\begin{proof}
This is the special case of Theorem~\ref{theorem:main1} when $\phi_\mathcal{B}(v)=v-\frac{1-F(v)}{f(v)}
=v-\frac{1}{\lambda}$. That gives us $\alpha=1, \beta=\frac{1}{\lambda}$. Following the fact that $\lim_{\alpha\rightarrow 1} \alpha^{\frac{1}{\alpha-1}}=e$ and $\lim_{\alpha\rightarrow 1} \alpha^{\frac{\alpha+1}{\alpha-1}}=e^2$, we prove the desired approximation factors. \qed
\end{proof}
\begin{corollary}[Power distributions] 
\label{coro:powerdistribution} 
Suppose $F(v)=1-(1-\frac{v}{\bar{v}})^a$ over the support $[0,\bar{v}]$ for some $a\geq 1$. Then the revenue of $\textrm{APX}(\frac{a+1}{a},\frac{\bar{v}}{a})$ (i.e. fee-setting with $w(P)=\frac{-1}{a}P+\frac{\bar{v}}{a}$) is $8-$approximation to the maximum surplus, 
and $4-$approximation to the maximum revenue.
\end{corollary}
\begin{proof}
This is the special case of Theorem~\ref{theorem:main1} when $\phi_\mathcal{B}(v)=v-\frac{1-F(v)}{f(v)}=\frac{a+1}{a}v-\frac{\bar{v}}{a}$. 
So $\alpha=\frac{a+1}{a}$, $\beta=\frac{\bar v}{a}$. Note that as $a\geq 1$, we have $\alpha\leq 2$. Following the fact that for $\alpha\leq 2$, we have $\alpha^{\frac{1}{\alpha-1}}\leq 4$  and $\alpha^{\frac{\alpha+1}{\alpha-1}}\leq 8$, we prove the desired approximation factors. \qed
\end{proof}
\subsection{Approximations for the uniform distribution}
Uniform distribution on $[0,1]$ is a special case 
of power distributions, so based on the results of the last section we can get approximation factors $4$ and $8$ with respect to optimal revenue and surplus respectively. However, we propose a different fee-setting mechanism that is $3-$approximation with respect to optimal revenue in expectation. Our technique is based on the ``best of two" technique for designing approximation algorithms, which picks the best of two mechanisms each performs well on some class of input seller's distribution. For the proof, see Appendix~\ref{proof:unif}.

\begin{theorem} 
\label{theorem:unif}Suppose $F=\textrm{unif}~[0,1]$. Let $y\triangleq\min \{\phi_\mathcal{S}^{-1}(1),\overline{c}\}$\footnote{ We set $\phi_\mathcal{S}^{-1}(1)=+\infty$ when $\phi_\mathcal{S}(1)=1$ doesn't have a solution.}. Then the mechanism which is best of $\textrm{APX}(2,1)$ and $\textrm{APX}(1,\frac{1-y}{2})$ in terms of revenue is $3$-approximation to optimal revenue in expectation. 
\end{theorem}

\begin{corollary}
The best affine fee-setting mechanism is at least a $3$-approximation to optimal revenue expectation when   $F=\textrm{unif}~[0,1]$.
\end{corollary}
\begin{proof}
The best affine fee-setting mechanism has expected revenue at least as large as both $\textrm{APX}(2,1)$ and $\textrm{APX}(1,\frac{1-y}{2})$, and hence is a $3-$approximation to the maximum intermediary's revenue.
\end{proof}

\subsection{Approximations for MHR distributions}
\label{sec:MHR}
In this section, we investigate the question of approximating surplus and revenue when neither buyer's virtual value nor seller's virtual cost is affine, but instead we have some proper distributional assumptions on the buyer and seller distributions. We look at the setting that the difference between the values of the seller and the buyer follows MHR distribution, which indicates that the surplus and revenue of an imaginary bidder with value $v-c$ are in constant approximation to each other. Moreover, we assume $v$ is coming from a MHR distribution and hence surplus approximation and revenue approximation are equivalent for this bidder. It is important to mention that many distributions in real economic exchange settings satisfy the following properties under independence assumption of seller and buyer (like uniform, normal, exponential, and etc.). Now, under these assumptions we get constant approximation ratio to both surplus and revenue in expectation with a constant fee schedule. Formally we have the following theorem.
\begin{theorem}
\label{theorem:MHR}
Suppose buyer's value $v$ is MHR, and random variable $v-c$ is also MHR. Then a constant fee-schedule mechanism $w(P)=\eta_{v-c}$ is $e^2-$approximation to optimal surplus, and hence $e^2-$approximation to optimal revenue in expectation, where $\eta_{v-c}$ is monopoly price of random variable $v-c$.
\end{theorem}
\begin{proof}
Let $P(c)$ be the BNE strategy of seller. We know random variable $v-c$ is MHR, hence due to Lemma~4.18 in \cite{hartline2012approximation}, monopoly revenue of $v-c$ is an $e-$approximation to maximum surplus of $v-c$ in expectation. In other words,
\begin{equation*}
R_\eta^{v-c}=\eta_{v-c}\textrm{Pr}\{v-c\geq \eta_{v-c}\}\geq \frac{1}{e}\mathbb{E}\{(v-c)_{+}\}=\frac{1}{e}\mathbb{E}\{\textrm{OPT-Surplus}\}.
\end{equation*}
 Moreover, the expected revenue of fee-setting mechanism $w(P)$ is equal to $\textrm{APX-Rev}=\eta_{v-c}\textrm{Pr}\{v\geq P(c)\}$. We claim $\mathbb{E}\{\textrm{APX-Rev}\}\geq \frac{1}{e} R_{\eta_{v-c}}$, which implies the desired approximation bounds. In other words,
\begin{equation*}
\mathbb{E}\{\textrm{APX-Rev}\}\geq \frac{1}{e} R_{\eta_{v-c}}\geq \frac{1}{e^2} \mathbb{E}\{\textrm{OPT-Surplus}\}\geq \frac{1}{e^2}\mathbb{E}\{\textrm{OPT-Rev}\}.
\end{equation*} 
To prove the claim, it is enough to show $\textrm{Pr}\{v\geq P(c)\}\geq\frac{1}{e} \textrm{Pr}\{v-c\geq \eta_{v-c}\}$. Note that from Corollary~\ref{affine:bnechar} we know $\phi_\mathcal{B}(P)=c+\eta_{v-c}$. Hence, conditioned on a fixed $c$ we have
\begin{align}
\label{eq:3}
\textrm{Pr}\{v-c\geq \eta_{v-c}|c\}=\textrm{Pr}\{v\geq c+\eta_{v-c}|c\}=1-F(\phi_\mathcal{B}(P))=e^{-H_\mathcal{B}(\phi_\mathcal{B}(P))}
\end{align}
Now, note that $H_{\mathcal{B}}(x)$ is convex (as $v$ is MHR), so $\forall x: H_{\mathcal{B}}(x)\geq H_{\mathcal{B}}(P)+h_{\mathcal{B}}(P)(x-P)$. Let $x=\phi_\mathcal{B}(P)$, and hence 
\begin{equation}
\label{eq:4}
H_\mathcal{B}(\phi_\mathcal{B}(P))\geq H_{\mathcal{B}}(P)+h_{\mathcal{B}}(P)(\phi_\mathcal{B}(P)-P)=H_{\mathcal{B}}(P)-1,
\end{equation}
where the last equality is true because $\phi_\mathcal{B}(P)=P-\frac{1-F(P)}{f(P)}=P-\frac{1}{h_{\mathcal{B}}(P)}$. Combining (\ref{eq:3}) and (\ref{eq:4}) we have 
\begin{equation}
\label{eq:5}
\textrm{Pr}\{v-c\geq \eta_{v-c}|c\}\leq e e^{-H_{\mathcal{B}}(P)}=e (1-F(P))=e \textrm{Pr}\{v\geq P|c\}.
\end{equation}
By taking expectation from both sides of (\ref{eq:5}) with respect to $c$ we prove what we claimed, which completes the proof of theorem.\qed
\end{proof}
\section{Inapproximability results}
\label{sec:inapprox}
In this section, we give two inapproximability results. The first one shows that
the proper fee schedules
  eBay and Amazon are currently using are not revenue-efficient,
in the sense that for $\textrm{unif}[0,1]$ buyer distribution no \textit{proper} fee schedule can get constant
approximation to the optimal revenue for the worst case seller distribution. Meanwhile, 
as we showed before, there is an improper fee-setting mechanism that always gets $4$-approximation to the optimal revenue. 
The second result 
shows that for $\textrm{unif}[0,1]$ buyer distribution, $\textrm{APX}(\alpha,\beta)$ 
gives seller prior independent constant approximation to the maximum surplus 
for worst-case seller distribution \textit{if and only if} $\alpha-\beta=1$
and $\alpha\neq 1$. 

\subsection{Inapproximability result for proper
fee schedule}
First we investigate the question of how good proper fee schedule works. We define 
a proper fee schedule as the following.
\begin{definition}
A proper fee schedule is an affine fee schedule with parameters $\alpha$ and $\beta$ such that $0\leq \alpha \leq 1$ and $\beta\geq 0$.
\end{definition}

Then we give definitions on the approximability of proper fee schedule.

\begin{definition}
Proper fee schedule revenue gap $RG_{F,G}$ under buyer distribution $F$,
and seller distribution $G$ is the ratio of 
the optimal revenue to 
the approximation revenue 
using the best proper fee schedule.
\end{definition}

\begin{definition}
Proper fee schedule surplus gap $SG_{F,G}$ under buyer distribution $F$,
and seller distribution $G$ is the ratio of 
the maximum surplus to 
the approximation revenue using the best proper fee schedule.
\end{definition}

As a direct consequence of Corollary~\ref{coro:powerdistribution}, we can say optimal revenue is $8-$approximation to optimal surplus in expectation. Hence, for the special case of $\textrm{unif[0,1]}$ we have,
\begin{corollary}\label{corollary:RGSG}
If $F$ is uniform distribution on $[0,1]$, then for any seller distribution
$G$, $RG_{F,G}\geq \frac{1}{8}SG_{F,G}$.
\end{corollary}

We now show the following theorem (proved in Appendix~\ref{inapprox:appendix}), 
which shows that $RG_{F,G}$ could be arbitrarily large
even if the buyer distribution is as simple as
the uniform $[0,1]$ distribution. At the same time, 
$\textrm{APX}(2,1)$ is $4$-approximation to the optimal revenue, 
which means proper fee schedule can be arbitrarily worse than 
$\textrm{APX}(2,1)$.
\begin{theorem} 
\label{inapprox:theorem}
When $F$ is uniform distribution on $[0,1]$,
for every constant $d$, there exists 
a regular seller distribution $G$ 
with $RG_{F,G}\geq d$. 
\end{theorem}
\proofsketch Based on Corollary \ref{corollary:RGSG}, 
it suffices to show that 
for every constant $d$, 
there exists a regular seller distribution $G$
with $SG_{F,G}\geq d$. 
Assume $F$ is uniform distribution on $[0,1]$. 
Consider the following 
family of distributions with parameter $\delta$ (as in Figure\ref{fig2}), defined on the interval $\left [0,1-\sqrt{\delta}\right ]$,
\begin{equation}
\label{instance:bad}
g_\delta(x)=\frac{2\delta}{(1-\delta)(1-x)^3}, 
~~~G_\delta(x)=\frac{\delta}{1-\delta}\left (\frac{1}{(1-x)^2}-1\right ), 
~~x\in \left [0,1-\sqrt{\delta}\right ].
\end{equation}
the rest of the proof shows that for any $d>0$, $\exists \delta$ such that $RG_{F,G}\geq d$\qed
\begin{figure}
\centering
\tikzstyle{letterlabel}=[thick,font=\fontsize{6}{6},color=black!70!blue]
\begin{tikzpicture}[scale=0.7]
\foreach \a in {0.1}
{	\draw[color=red,domain=0:4.10263] plot (\x,{6*\a/(1-\a)*(1/(1-\x/6)^2-1)}) node[above] {$F_{0.1}$};
}
\foreach \a in {0.1}
{	\draw[color=red,domain=0:2.6,dashed] plot (\x,{6*2*\a/(1-\a)/((1-\x/6)^3)        }) node[above] {$f_{0.1}$};
}
\foreach \a in {0.01}
{	\draw[color=green,domain=0:5.4] plot (\x,{6*\a/(1-\a)*(1/(1-\x/6)^2-1)}) node[above] {$F_{0.01}$};
}
\foreach \a in {0.01}
{	\draw[color=green,domain=0:4.47,dashed] plot (\x,{6*2*\a/(1-\a)/((1-\x/6)^3)        }) node[above] {$f_{0.01}$};
}
\foreach \a in {0.001}
{	\draw[color=blue,domain=0:5.810263,samples=400] plot (\x,{6*\a/(1-\a)*(1/(1-\x/6)^2-1)}) node[right] {$F_{0.001}$};
}
\foreach \a in {0.001}
{	\draw[color=blue,domain=0:5.3,samples=400,dashed] plot (\x,{6*2*\a/(1-\a)/((1-\x/6)^3)   }) node[right] {$f_{0.001}$};
}
\draw (-0.2,6) node {$1$};
\draw[densely dotted,color=red!60] (4.10263,0) -- + (0,8);
\draw[densely dotted,color=green!60] (5.4,0) -- + (0,8);
\draw[densely dotted,color=blue!60] (5.810263,0) -- + (0,8);
\draw[densely dotted] (-1,6) -- (6,6);
\draw[densely dotted] (6,-1) -- (6,6);
\draw (6,-0.3) node {$1$};
\draw (-0.2,-0.3) node {$0$};
\draw (-0.4,0.2222222*6) node {$2/9$};
\draw (-0.45,0.02020202*6+0.05) node {$2/99$};
\draw[->,very thick] (-1,0) -- (7,0);
\draw[->,very thick] (0,-1) -- (0,7);
\end{tikzpicture}
\caption{Family of worst-case seller distributions used in Theorem~\ref{inapprox:theorem}}
\label{fig2}
\end{figure}
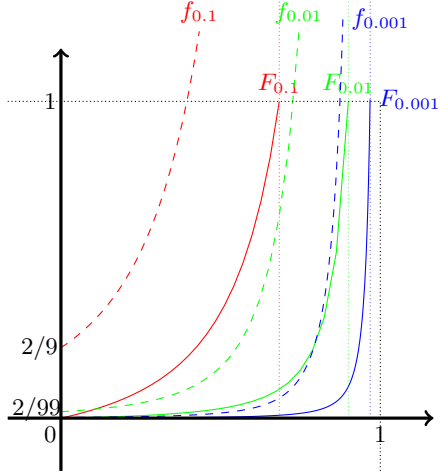
\subsection{Inapproximability result for prior-independent approximation}
For the setting of seller prior-independent, one might still expect the existence of other constant approximations. However, we show our mechanism is the unique fee-setting mechanism that can get constant seller prior-independent approximations to surplus. More formally, we show that in the seller  prior-independent setting when buyer's value is drawn from uniform $[0,1]$ distribution, $w(P)=(1-\alpha)P +\beta$ gives
constant approximation to 
the surplus if and only if $\alpha-\beta=1$. The proof is provided in Appendix~\ref{inapprox:appendix}.
\begin{theorem}
\label{theorem:priorindinapprox}
If the buyer's distribution is uniform $[0,1]$, 
$w(x)=(1-\alpha)x+\beta$, where $\alpha$ and $\beta$ are parameters independent form the seller distribution, then 
the revenue obtained using $w$ is a constant approximation to the 
surplus for every possible
seller's distribution if and only if $\alpha-\beta=1$
and $\alpha\neq 1$. 
Moreover, when $\alpha=2, \beta=1$, 
it achieves the best approximation ratio $8$. 
\end{theorem}
\section{Extension to multi-buyers case}
Some of our results can extend to multi-buyers case, when the buyers are regular and i.i.d. In fact, if there are $n$ buyers with regular i.i.d. values $v_1,v_2,\ldots,v_n$ drawn from distribution $F$, one can replace the pool of buyers with one \textit{effective} buyer $v=\underset{i}{\max}~v_i$ and still get the same revenue in expectation for any fee-setting mechanisms and optimal Myerson mechanism (because all buyers have the same non-decreasing virtual value function), and also the same surplus in expectation for VCG mechanism. Now, using the following lemma and the above reduction we can extend Theorem~\ref{theorem:MHR} to multi-buyers case, whose proof is found in Appendix~\ref{lemmaproofs:appendix}.
\begin{lemma}
\label{iidMHR}
Suppose $v_1,v_2,\ldots,v_n$ are i.i.d. random variables drawn from MHR distribution $F$. Then $v=\underset{i}{\max}~v_i$ is also MHR.
\end{lemma}
Now, by combining Lemma~\ref{iidMHR} and Theorem~\ref{theorem:MHR} we have the following direct corollary.
\begin{corollary}[Multi-buyers setting] Suppose $F$ is a MHR distribution and there are $n$ i.i.d. buyers whose values are drawn from $F$. Seller's cost $c$ is drawn from $G$ and is independent from all buyers. Moreover, assume the random variable $\underset{i}{\max}~v_i-c$ is MHR. Then the revenue of constant fee-setting mechanism $w(P)=\eta_{\underset{i}{\max}~v_i -c}$, where $\eta_{\underset{i}{\max}~v_i-c}$ is the monopoly price for the distribution of $\underset{i}{\max}~v_i -c$, is $e^2-$approximation to optimal surplus and revenue in expectation.
\end{corollary}
\section*{Acknowledgment}
The authors would like to express their very great appreciations to Prof. Jason Hartline for his valuable and constructive suggestions during this research work, especially in developing Theorem~\ref{theorem:MHR}. His willingness to give his time so generously has been very much appreciated.

\appendix
\bibliography{ebay}
\bibliographystyle{plainnat}
\section{Basics of mechanism design for exchange}
\setcounter{section}{1}
\label{mechbasic}
In this section, we provide details of solution concepts and definitions used in this paper. The provided details are following the normal trend of mechanism design literature, but they have been adapted for the exchange setting.

 Similar to the definition of allocations and payments in single dimensional mechanism design framework,  suppose $x_\mathcal{S}\in\{0,1\}$ and $p_\mathcal{S}\in[0,\infty)$ are seller's allocation and payment, and $x_\mathcal{B}\in\{0,1\}$ and $p_\mathcal{B}\in[0,\infty)$ are buyer's allocation and payment. In the context of exchange, feasible allocations are $(x_\mathcal{B},x_\mathcal{S})=\{(1,1),(0,0)\}$.  We assume both seller and buyer are risk-neutral, i.e. $u_\mathcal{S}=p_\mathcal{S}-x_s c$ and $u_\mathcal{B}=x_B v-p_\mathcal{B}$ are seller's and buyer's utilities under allocations $\mathbf{x}=(x_\mathcal{B},x_\mathcal{S})$  and payments $\mathbf{p}=(p_\mathcal{B},p_\mathcal{S})$. We start by defining an \textit{exchange mechanism} when we have one seller and one buyer as follows.

\begin{definition} An \textit{exchange mechanism} for  1-seller, 1-buyer is a tuple $\mathscr{M}=(A_{\mathcal{S}}\times A_{\mathcal{B}},\mathbf{x}(.),\mathbf{p}(.) )$, where  $\mathbf{x}=(x_\mathcal{S}(.) ,x_\mathcal{B}(.))$, and $\mathbf{p}= (p_\mathcal{S}(.) ,p_\mathcal{B}(.))$. $A_\mathcal{S}$ and $A_\mathcal{B}$ are set of mechanism actions of seller and buyer respectively, $x_\mathcal{S}:A_{\mathcal{B}}\times A_{\mathcal{S}}\rightarrow \{0,1\}$ and $x_\mathcal{B}:A_{\mathcal{B}}\times A_{\mathcal{S}}\rightarrow \{0,1\}$ are \textit{seller's allocation}  and \textit{buyer's allocation} respectively, and $p_\mathcal{S}:A_{\mathcal{S}}\rightarrow [0,\infty)$ and $p_\mathcal{B}:A_{\mathcal{B}}\rightarrow [0,\infty) $ are \textit{seller's payment}  and \textit{buyer's payment} respectively. Moreover, mechanism $\mathscr{M}$ and strategy profile $\left(b_\mathcal{B}(.),b_{\mathcal{S}}(.)\right)$ where $b_\mathcal{B}:[0,\bar{v})\rightarrow \mathcal{A}_{\mathcal{B}}, b_\mathcal{S}:[0,\bar{c})\rightarrow \mathcal{A}_{\mathcal{S}}$ implement \emph{allocation rules} $x_\mathcal{S}(c,v)\triangleq x_\mathcal{S}(b_{\mathcal{S}}(c),b_\mathcal{B}(v))$, $x_\mathcal{B}(c,v)\triangleq x_\mathcal{B}(b_{\mathcal{S}}(c),b_\mathcal{B}(v))$, and \emph{payment rules} $p_\mathcal{S}(v,c)\triangleq p_\mathcal{S}(b_{\mathcal{S}}(c),b_\mathcal{B}(v))$, $p_\mathcal{B}(v,c)\triangleq p_\mathcal{B}(b_{\mathcal{S}}(c),b_\mathcal{B}(v))$. 
\end{definition}
\begin{definition}
Suppose we have an exchange mechanism $\mathscr{M}$ and a strategy profile $(b_\mathcal{S},b_\mathcal{B})$ that implement allocation/payment rules  $(x_\mathcal{S}(v,c),x_{\mathcal{B}}(v,c))$ and $(p_\mathcal{S}(v,c),p_{\mathcal{B}}(v,c))$ respectively. Now, \emph{interim allocation rule} (or \emph{interim payment rule}) of an agent (either seller or buyer) is defined as the expectation of the allocation rule (or payment rule) of that agent conditioned on her private information (cost if seller, value if buyer). We denote interim allocation/payment rules by $(x_\mathcal{S}(c),x_\mathcal{B}(v))$ and $(p_\mathcal{S}(c),p_\mathcal{B}(v))$ by a bit of notation abuse (as allocation/payment rules use the same notation as interim allocation/payment rules but with different function inputs).
\end{definition}

Bayesian mechanism design in general aims to define the rules of a game of incomplete information, a.k.a. the mechanism,  played by the agents in the environment. Mechanism designer hopes that a solution of this game has desirable properties, in particular good objective functions such as revenue of the mechanism or surplus of the agents. To analyze the solution of the game, we need to look at the correct solution concept applicable to our application. To do so, we first formalize the game which is played by seller and buyer, and then we talk about solution concepts we use in this paper. As it can be seen from the above definition, a \textit{strategy} for an agent (buyer or seller) is a mapping from its type space (i.e. value space of the buyer or cost space of the seller) to its corresponding mechanism's action space ($A_\mathcal{S}$ for seller or $A_\mathcal{B}$ for buyer). We now define a \textit{direct revelation exchange mechanism}.
\begin{definition}
A \textit{direct revelation exchange mechanism} is a single-round, sealed bid exchange mechanism which has action spaces equal to the corresponding type spaces (i.e., the seller bids its cost for the item under trade and the buyer bids its value for that item).
\end{definition}
We now can define a Bayes-Nash Equilibrium strategy profile of an exchange mechanism as follows.
\begin{definition}
A \emph{Bayes-Nash Equilibrium} for an exchange mechanism $\mathscr{M}=(A_{\mathcal{S}}\times A_{\mathcal{B}},\mathbf{x}(.),\mathbf{p}(.) )$ under common prior $F\times G$ is a strategy profile $\left(b_\mathcal{B}(.),b_{\mathcal{S}}(.)\right)$ where $b_\mathcal{B}:[0,\bar{v})\rightarrow \mathcal{A}_{\mathcal{B}}, b_\mathcal{S}:[0,\bar{c})\rightarrow \mathcal{A}_{\mathcal{S}}$, and 
\begin{align*}
&-\forall v\in [0,\bar{v}), \forall ~b'_{\mathcal{B}}: \\&~~\mathbb{E}_{c}\{u_{\mathcal{B}}[x_\mathcal{B}(b_{\mathcal{B}}(v),b_{\mathcal{S}}(c)) ,p_\mathcal{B}(b_{\mathcal{B}}(v),b_{\mathcal{S}}(c))]\}\geq \mathbb{E}_{c}\{u_{\mathcal{B}}[x_\mathcal{B}(b'_{\mathcal{B}}(v),b_{\mathcal{S}}(c)) ,p_\mathcal{B}(b'_{\mathcal{B}}(v),b_{\mathcal{S}}(c))]\}
\end{align*}
\vspace{-18pt}
\begin{align*}
&-\forall c\in [0,\bar{c}), \forall ~b'_{\mathcal{S}}: \\&~~\mathbb{E}_{v}\{u_{\mathcal{S}}[x_\mathcal{S}(b_{\mathcal{B}}(v),b_{\mathcal{S}}(c)) ,p_\mathcal{S}(b_{\mathcal{B}}(v),b_{\mathcal{S}}(c))]\}\geq \mathbb{E}_{v}\{u_{\mathcal{S}}[x_\mathcal{S}(b_{\mathcal{B}}(v),b'_{\mathcal{S}}(c)) ,p_\mathcal{S}(b_{\mathcal{B}}(v),b'_{\mathcal{S}}(c))]\}
\end{align*}

\end{definition}
Similar to the definition of BNE, we adapt the solution concepts of \emph{Bayesian Incentive Compatibility (BIC)}, \emph{Dominant Strategy Incentive Compatibility (DSIC)}, \emph{Interim Individual Rationality (Interim IR)}, and \emph{Ex-post Individual Rationality (Ex-post IR)} to the setting of exchange as follows.
\begin{definition}
A direct revelation exchange mechanism $\mathscr{M}$ is \emph{BIC} if truthful bidding (i.e. seller bids her cost, buyer bids her value) is a BNE.
\end{definition}
\begin{definition}
An exchange mechanism $\mathscr{M}$ is \emph{Interim IR} if neither buyer nor seller get a negative revenue in expectation at the interim stage of the game, i.e. when they know their private types.
\end{definition}
\begin{definition}
An exchange mechanism $\mathscr{M}$ is \emph{Ex-post} IR if neither buyer nor seller get a negative revenue at the ex-post stage of the game, i.e. when all the private types are revealed to all the players.
\end{definition}

Below We look at the intermediary problem as a single dimensional mechanism design framework, and  characterize the optimal Myerson's mechanism for single item exchange problem. We implement this optimal revenue scheme using a more intuitive indirect mechanism in Section~\ref{indopt}. This second indirect mechanism will be the base-line of all of our proposed simple mechanisms in this paper. 
\begin{theorem}[Myerson's mechanism for exchange]
Suppose both seller and buyer are regular. Then the following direct BIC and Interim IR mechanism\footnote{To be more precise, this mechanism is also Dominant Strategy Incentive Compatible (DSIC) and Ex-post individual rational under no-positive transfer assumption.}, which is maximizing \emph{virtual surplus} ($\triangleq x_\mathcal{B}\phi_\mathcal{B}- x_\mathcal{S}\phi_\mathcal{S}$), is revenue optimal in expectation.
\begin{itemize}
\item Solicit seller's and buyer's bids for their cost and value respectively. Let these bids be ($b_\mathcal{S},b_{\mathcal{B}})$.
\item If $\phi_\mathcal{B}(b_\mathcal{B})\geq \phi_\mathcal{S}(b_\mathcal{S})$ the trade happens, o.w. no item will be transferred. 
\item If trade happens, then charge the buyer its critical price, i.e. $\tau_{B}=\phi_\mathcal{B}^{-1}(\phi_\mathcal{S}(b_\mathcal{S}))$, and give the seller its critical price, i.e. $\tau_{S}=\phi_\mathcal{S}^{-1}(\phi_\mathcal{B}(b_\mathcal{B}))$. Otherwise, nobody will be charged.
\end{itemize}
\end{theorem}
Other than the objective of revenue, another important benchmark in mechanism design is \textit{surplus}. Vickrey-Clarke-Groves (VCG) mechanism maximizes the surplus and satisfies the strongest incentive compatibility and individual rationality solution concepts. We adapt VCG mechanism to the setting of exchange as follows.
\begin{theorem}
The following  DSIC and ex-post IR mechanism is maximizing the  surplus:
\begin{itemize}
\item Solicit seller's and buyer's bids for their cost and value respectively. Let these bids be ($b_\mathcal{S},b_{\mathcal{B}})$.
\item  If $b_\mathcal{B}\geq b_\mathcal{S}$ the trade happens, o.w. no item will be transferred. 
\item If trade happens, then pay the seller the amount of $b_\mathcal{B}$ and charge the buyer  $b_\mathcal{S}$. 
\end{itemize}
\end{theorem}

Finally, we need to define the notion of Prior-Independence with respect to seller market or buyer market for exchange mechanisms.
\begin{definition}
An exchange mechanism $M$ is known to be seller  prior-independent with respect to the seller if no seller distributional information is needed by the mechanism. 
\end{definition}
\section{Implantation of Myerson's optimal by fee-setting}
\label{appsec3}
One important question related to designing fee-setting mechanisms is whether they can be optimal or not.
 The following theorem, proved first in \cite{LN07}, provides an answer to this question. It states with a proper choice of function $w(P)$ one can design a fee-setting mechanism that extracts the same revenue in expectation as in Myerson's optimal mechanism. We provide a simple proof of this result using Revenue Equivalence theorem~\cite{RB78}.
\begin{theorem}\cite{LN07}
Consider an exchange setting with regular buyer/seller. Define $\mathscr{P}(c)\triangleq\phi_{\mathcal{B}}^{-1}(\phi_{\mathcal{S}}(c))$.  Consider a fee-setting exchange mechanism with fee schedule $w(P)=P-\mathbb{E}_v\{\mathscr{P}^{-1}(v)|v\geq P\}$. Then:
\begin{itemize}
\item $P(c)=\mathscr{P}(c)$ and $b(v)=v$ is a BNE of this mechanism.
 \item The interim allocation/payment rules are equal to those of Myerson's optimal mechanism.
\end{itemize}
\end{theorem}
\begin{proof}
Suppose seller plays $p(c)=\mathscr{P}(c)=\phi_{\mathcal{B}}^{-1}(\phi_{\mathcal{S}}(c))$ and buyer plays $b(v)=v$.  Now, let $(x_{\mathcal{S}}(v,c), x_{\mathcal{B}}(v,c)) $ be the interim allocation rule  and $(p_{\mathcal{S}}(v,c), p_{\mathcal{B}}(v,c)) $ the interim payment rule of this mechanism under this strategy profile. Moreover, let $(x^M_{\mathcal{S}}(v,c), x^M_{\mathcal{B}}(v,c)) $ and $(p_{\mathcal{S}}(v,c), p_{\mathcal{B}}(v,c)) $ be the interim allocation rule and interim payment rule of Myerson's optimal mechanism respectively (Note that we know Myerson's is DSIC). In this mechanism, trade happens when $v \geq \phi_{\mathcal{B}}^{-1}(\phi_{\mathcal{S}}(c))$ which is equivalent to $\phi_{\mathcal{B}}(v)\geq\phi_{\mathcal{S}}(c)$. So both of our mechanism and Myerson's optimal mechanism have the same allocation rule for both buyer and seller, i.e. $x_\mathcal{B}=x^M_\mathcal{B}=x_\mathcal{S}=x^M_\mathcal{S}=\mathds{1}\{\phi_{\mathcal{B}}(v)\geq\phi_{\mathcal{S}}(c)\}$. As in the Myerson's mechanism the critical price of the buyer is $\tau_\mathcal{B}=\phi_{\mathcal{B}}^{-1}(\phi_{\mathcal{S}}(c))$ and we charge the buyer by $\mathscr{P}(c)$ if trade happens, the payment rule of buyers are the same in both mechanism. For the interim payment rule of the seller in our mechanism we have $p_{\mathcal{S}}(c)=\mathbb{E}_v\{\left(p(c)-w(p(c))\right)\mathds{1}\{v\geq p(c)\}\}=(p(c)-w(p(c)))(1-F(p(c))$. In the Myerson's mechanism, we have $p^M_{\mathcal{S}}(c)=\mathbb{E}_v\{\tau_\mathcal{S}(v)\mathds{1}\{c\leq \tau_\mathcal{S}(v)\}\}=\mathbb{E}_v\{\mathscr{P}^{-1}(v)|v\geq p(c)\}\textrm{Pr}_v\{v\geq p(c)\}=(p(c)-w(p(c))(1-F(p(c)))=p_{\mathcal{S}}(c)$. Hence both mechanisms have the same interim allocation/payment rules. As Myerson's mechanism is BIC, we conclude that $p(c)=\mathscr{P}(c)$ and $b(v)=v$ are also BNE of our mechanism due to revenue equivalence theorem~\cite{RB78}.
\end{proof}
\begin{corollary}
The indirect fee-setting exchange mechanism with fee schedule $w(P)=\mathbb{E}_v\{\mathscr{P}^{-1}(v)|v\geq P\}$ extracts the maximum revenue in expectation under BNE strategy profile $(\mathscr{P}(c), v)$ for seller and buyer.
\end{corollary}

For the special case when the seller's virtual cost is affine, there is an interesting result due to \citet{LN13} which shows the fee-setting mechanism that implements the optimal Myerson is also affine. More formally, we have the following theorem (modified a bit) due to \citet{LN13}.
\begin{theorem} \cite{LN13} Suppose the buyer is buyer-regular. Then the following are equivalent statements:
\begin{itemize}
\item Cost of the seller is drawn from a reverse-generalized Pareto distribution with parameters $\mu,\lambda$ and $\xi$. 
\item An affine fee mechanism, i.e. with fee schedule $w(P)=(1-\alpha)P+\beta$ where $\alpha=\frac{1}{1+\xi}$ and $\beta=-\frac{\frac{1}{\lambda}+\xi\mu}{1+\xi}$,  is intermediary optimal for all buyer distributions.
\end{itemize}
\end{theorem}
\begin{remark}
In principle, one can run the dual of a fee-setting mechanism by swapping the roles of buyer and seller: mechanism asks the buyer for a price $P$ and post the price for the seller. If seller is willing to sell the item with price $P$, intermediary takes $w(P)$ as its share and charges the buyer by $w(P)+P$. Although guarantee bounds for these mechanisms are equivalent to those of ordinary fee-settings,  these fee-setting mechanisms are often not used  by online exchange platforms in reality, such as Amazon or eBay. Hence, they are out of focus of this paper.
\end{remark}
\section{Details of the proof of Theorem~\ref{theorem:main1}}
\label{detail:proof}
\begin{proof}
We showed the following relations while sketching the proof
of Theorem~\ref{theorem:main1} in Section~\ref{mainresult}.
\begin{align}
&\mathbb{E}\{\textrm{Rev-APX}|c\}=\frac{w(P)^{\frac{\alpha}{\alpha-1}}}{\beta^{\frac{1}{\alpha-1}}}\label{eq:first}\\
&\mathbb{E}\{\textrm{OPT-Rev}|c\}\leq \alpha\mathbb{E}\left\{w(v)\mathds{1}\{v\geq \phi_\mathcal{B}(P)\}|c\right\}\label{eq:sec}\\
&\mathbb{E}\{\textrm{OPT-Surplus}|c\}=\mathbb{E}\{w(v)\mathds{1}\{v\geq \phi_\mathcal{B}(\phi_\mathcal{B}(P))\}|c\}\label{eq:third}
\end{align}
Now, we find equivalent expressions for upper-bounds in (\ref{eq:first}) and (\ref{eq:sec}) as follows.
\begin{align}
&\alpha\mathbb{E}\left\{w(v)\mathds{1}\{v\geq \phi_\mathcal{B}(P)\}|c\right\}=\alpha \int_{t\geq \phi_\mathcal{B}(P)}(1-F(t))dt=\alpha \int_{t\geq \phi_\mathcal{B}(P)}\left (\frac{w(t)}{\beta}\right )^{\frac{1}{\alpha-1}}dt\nonumber\\
&=\frac{1}{\alpha\beta^{\frac{1}{\alpha-1}}}w(\phi_\mathcal{B}(P))^{\frac{\alpha}{\alpha-1}}=\alpha^{\frac{1}{\alpha-1}} \frac{w(P)^{\frac{\alpha}{\alpha-1}}}{\beta^{\frac{1}{\alpha-1}}},
\end{align}
where in the last equality we use the fact that $w(\phi_\mathcal{B}(P))=\alpha w(P)$, due to property (a) in Lemma~\ref{affine:property}. Also, we have
\begin{align}
&\alpha\mathbb{E}\left\{w(v)\mathds{1}\{v\geq \phi_\mathcal{B}(\phi_\mathcal{B}(P))|c\right\}=\alpha \int_{t\geq \phi_\mathcal{B}(\phi_\mathcal{B}(P))}(1-F(t))dt=\alpha \int_{t\geq \phi_\mathcal{B}(\phi_\mathcal{B}(P))}\left (\frac{w(t)}{\beta}\right )^{\frac{1}{\alpha-1}}dt\nonumber\\
&=\frac{1}{\alpha\beta^{\frac{1}{\alpha-1}}}w(\phi_\mathcal{B}\left (\phi_\mathcal{B}(P))\right )^{\frac{\alpha}{\alpha-1}}=\alpha^{\frac{\alpha+1}{\alpha-1}} \frac{w(P)^{\frac{\alpha}{\alpha-1}}}{\beta^{\frac{1}{\alpha-1}}},
\end{align}
where in the last equality we use the fact that $w(\phi_\mathcal{B}(\phi_\mathcal{B}(P)))=\alpha^2 w(P)$, due to property (a) in Lemma~\ref{affine:property}. Comparing the above upper-bounds on optimal revenue and surplus with the revenue of the affine fee-setting mechanism given in~\eqref{eq:first} completes the proof of the desired approximation factors.\qed
\end{proof}
\section{Proof of Theorem~\ref{theorem:unif}}
\label{proof:unif}
\begin{proof}
For uniform $[0,1]$ distribution, 
$F(x)=x$, $f(x)=1$, $\phi(x)=2x-1$. 
So $\alpha=2,\beta=1, P(c)=\frac{c+3}{4}$.
We first derive an upperbound on $\textrm{OPT}$. We have
\begin{align}
\textrm{OPT}
&=\int_{c=0}^{y}\left(\int_{0.5+0.5\phi_\mathcal{S}(c)}^{1}(2v-1-\phi_\mathcal{S}(c))dv\right) g(c)dc\nonumber\\
&=\int_{c=0}^{y} \Big( 1-(0.5+0.5\phi_\mathcal{S}(c))^2-0.5(1+\phi_\mathcal{S}(c))(0.5-0.5\phi_\mathcal{S}(c))\Big) g(c)dc\nonumber\\
&=\frac{1}{4}\int_{c=0}^{y}\Big(1-\phi_\mathcal{S}(c)\Big )^2 g(c)dc=\frac{1}{4}\int_{c=0}^{y}\left(1-c-\frac{G(c)}{g(c)}\right)^2 g(c)dc\nonumber\\
&=\frac{1}{4}\int_{c=0}^{y}\left(1+c^2+\frac{G^2(c)}{g^2(c)}-2c-2\frac{G(c)}{g(c)}+2c\frac{G(c)}{g(c)}\right)g(c)dc\nonumber\\
&=\frac{1}{4}\left( G(y)+y^2G(y)-2\int_{c=0}^y cG(c)dc -2yG(y)+2\int_{c=0}^{y}G(c)dc- 2\int_{c=0}^{y}cG(c)dc+2\int_{c=0}^{y}G(c)dc\right)
\nonumber\\
&+\frac{1}{4}\int_{c=0}^{y}\left(\frac{G(c)}{g(c)}\right)^2g(c)dc=\frac{(1-y)^2G(y)}{4}+\frac{1}{4}\int_{c=0}^{y}\left(\frac{G(c)}{g(c)}\right)G(c)dc\nonumber\\
&\leq \frac{(1-y)^2G(y)}{4}+\frac{1}{4}\int_{c=0}^{y}\left(1-c\right)G(c)dc,
\end{align}

where the last inequality comes from the fact that due to seller regularity, for $c\in[0,y]:\phi_\mathcal{S}(c)\leq 1\Rightarrow \frac{G(c)}{g(c)}\leq (1-c)$. Now we have
\begin{align}
\textrm{OPT}&\leq \frac{(1-y)^2G(y)}{4}+\frac{1}{4}\int_{c=0}^{y}\left(1-c\right)G(c)dc\nonumber\\
&=\frac{(1-y)^2G(y)}{4}-\frac{1}{8}(1-c)^2G(c)\bigr|_{c=0}^{c=y}+\frac{1}{8}\int_{c=0}^{y}(1-c)^2g(c)dc\nonumber\\
&=\frac{(1-y)^2G(y)}{8}+\frac{1}{8}\int_{c=0}^{y}(1-c)^2g(c)dc.
\end{align}
Let $\textrm{OPT}_1\triangleq \frac{(1-y)^2G(y)}{8}$ and $\textrm{OPT}_2\triangleq \frac{1}{8}\int_{c=0}^{y}\left(1-c\right)^2g(c)$. We now show $\textrm{Rev-APX}(1,\frac{1-y}{2})\geq \textrm{OPT}_1 $ and $\textrm{Rev-APX}(2,1)\geq \frac{\textrm{OPT}_2}{2} $, and hence conclude the best of these two mechanisms is always a $3$-approximation to $\textrm{OPT}$. To show this, we look at the exact expression for $\textrm{Rev-APX}(\alpha,\beta)$. 
\begin{align}
\textrm{Rev-APX}(\alpha,\beta)&=\mathbb{E}_{c,v}\{w(P(c)\mathds{1}\{v\geq P(c)\}\}\nonumber\\
&=\int_{c=0}^{\alpha-\beta}\left(\frac{(1-\alpha)(c+\alpha)+\beta(1+\alpha)}{2\alpha}\right)\left( \frac{\alpha -c -\beta}{2\alpha}\right)g(c)dc\nonumber\\
&=\frac{1}{4\alpha^2}\int_{c=0}^{\alpha-\beta}\Big((1-\alpha)(c+\alpha)+\beta(1+\alpha)\Big) \left(\alpha-c-\beta\right)g(c)dc.
\label{eq3}
\end{align}
Now, we first find a lower bound for $\textrm{Rev-APX}(1,\frac{1-y}{2})$. By applying integration by parts and using the fact that $G(.)$ is monotone non-decreasing we have
\begin{align}
\textrm{Rev-APX}(1,\frac{1-y}{2})&=\frac{1}{4}\int_{c=0}^{\frac{1+y}{2}}(1-y)\left (\frac{1+y}{2}-c\right )g(c)dc\nonumber\\
&=\frac{1-y}{4}\left( \left (\frac{1+y}{2}-c\right )G(c)\bigr|_{0}^{\frac{1+y}{2}}\right)+\frac{1-y}{4}\int_{c=0}^{\frac{1+y}{2}}G(c)dc\nonumber\\
&=\frac{1-y}{4}\int_{c=0}^{\frac{1+y}{2}}G(c)dc\geq\frac{1-y}{4} \int_{c=y}^{\frac{1+y}{2}}G(c)dc\geq \frac{1-y}{4} \int_{c=y}^{\frac{1+y}{2}}G(y)dc\nonumber\\
&=\frac{(1-y)^2G(y)}{8}=\textrm{OPT}_1.
\end{align}
Note that in the above calculation, $y=\min \{\phi_\mathcal{S}^{-1}(1),\overline{c}\}\leq 1$, as $\phi_\mathcal{S}(1)\geq 1$, and hence $\frac{1+y}{2}\geq y$. 
Based on the previous calculation, we know $\textrm{Rev-APX}(2,1)\geq \frac{1}{16}\int_{c=0}^{y}(1-c)^2g(c)dc=\frac{\textrm{OPT}_2}{2}$, which completes the proof.
\end{proof}
\section{Proof of inapproximibility results}
\label{inapprox:appendix}
\begin{proof}[Proof of Theorem~\ref{inapprox:theorem}]Consider the family of seller distributions proposed in (\ref{instance:bad}). First of all, for this family of distributions we have $\phi_\mathcal{S}(c)=c+\frac{G_\delta(c)}{g_\delta(c)}=c+\frac{\frac{1}{(1-x)^2}-1}{\frac{2}{(1-x)^3}}=\frac{1+x}{2}-\frac{(1-x)^3}{2}$ is a non-decreasing function and hence seller is regular. Next step to prove the theorem is coming up with an expression for maximum social surplus in terms of parameter $\delta$. We have,
\begin{align}
\textrm{Max-Surplus}_\delta&=\frac{1}{2}\mathbb{E}_c\left \{(1-c)^2\right \}\nonumber
=\frac{\delta}{1-\delta}\int_0^{1-\sqrt\delta}\frac{(1-c)^2}{(1-c)^3}dc=\frac{\delta}{1-\delta}\ln\frac{1}{\sqrt \delta}=\frac{\delta}{2(1-\delta)}\ln\frac{1}{\delta}.
\end{align}
Let $\alpha-\beta=1-\epsilon$.
Define $\xi\triangleq \max(\sqrt\delta,\epsilon)$.
Based on (\ref{eq3}), for every possible
pair $(\alpha,\beta)$ under the distribution $G_\delta$,
we have,
\begin{align}
&\textrm{APX}_\delta(\alpha, \alpha-1+\epsilon)=\frac{1}{4\alpha^2}\mathbb{E}_c\left\{\Big((\alpha-1)(1-c)+\epsilon(\alpha+1)\Big)\left(1-c-\epsilon\right)_+\right\}
\\
=&\frac{1}{4\alpha^2}\int_{0}^{\min(1-\sqrt{\delta},1-\epsilon)}{\Big((\alpha-1)(1-c)+\epsilon(\alpha+1)\Big)\left(1-c-\epsilon\right)g_\delta(c)dc}\nonumber\\
=&\frac{\alpha-1}{4\alpha^2}\int_{0}^{1-\xi}(1-c)^2g_\delta(c)dc+\frac{\epsilon}{2\alpha^2}\int_{0}^{1-\xi}(1-c)g_\delta(c)dc-\frac{\epsilon^2(\alpha+1)}{4\alpha^2}\int_{0}^{1-\xi}g_\delta(c)dc.\nonumber
\end{align}
By plugging $g_\delta(c)=\frac{2\delta}{(1-\delta)(1-c)^3}$ for $c\in [0,1-\xi]$ and computing integrals we have
\begin{align}
\textrm{APX}_\delta(\alpha, \alpha-1+\epsilon)=\frac{(\alpha-1)\delta}{2\alpha^2(1-\delta)}\ln \frac{1}{\xi}+\frac{\delta\epsilon}{\alpha^2(1-\delta)}\frac{(1-\xi)}{\xi}-\frac{\epsilon^2(\alpha+1)\delta}{4\alpha^2(1-\delta)}\frac{(1-\xi^2)}{\xi^2}.
\end{align}

Now we have the expressions for both $\textrm{Surplus}$ and 
$\textrm{APX}_\delta(\alpha, \beta)$. 
Below we want to show, if $\delta$ is small enough, 
$SG_{F,G}$ under $G_\delta$ can be arbitrarily large, which means
for every constant $d$, we can find a $G_\delta$ with $SG_{F,G}\geq d$. 
In order to prove this, for fixed $\delta$ we discuss the possible
values of $(\alpha,\epsilon)$, and then compute the ratio of
$\textrm{APX}_\delta(\alpha, \beta)$
to $\textrm{Surplus}$. If the ratio goes to zero as 
$\delta$ goes to zero, the theorem is proved.

We now consider two cases:
\begin{itemize}
\item \textit{Case 1} ($\epsilon\leq \sqrt{\delta}$): In this case we have $\xi=\sqrt{\delta}$, and hence
\begin{equation}
\textrm{APX}_\delta(\alpha, \alpha-1+\epsilon)=\frac{(\alpha-1)\delta}{4\alpha^2(1-\delta)}\ln \frac{1}{\delta}+\Delta_\delta(\epsilon),
\end{equation}
where $\Delta_\delta(\epsilon)\triangleq\frac{\sqrt\delta(1-\sqrt \delta)}{\alpha^2(1-\delta)}\epsilon-\frac{(\alpha+1)}{4\alpha^2}\epsilon^2$. This function is quadratic with respect to its argument and its maximum over the interval $[0,\sqrt\delta]$ happens at either  $\epsilon^*=\frac{2\sqrt\delta(1-\sqrt \delta)}{(\alpha+1)(1-\delta)}$ or $\sqrt\delta$. Now, we develop an upperbound on the ratio of $\textrm{APX}_\delta(\alpha,\alpha-1+\epsilon)$ and $\textrm{Max-Surplus}$ when $\epsilon\leq \sqrt\delta$ as follows:
\begin{equation}
\label{ineq1}
\forall \epsilon\in[0,\sqrt \delta]: \frac{\textrm{APX}_\delta(\alpha,\alpha-1+\epsilon)}{\textrm{Max-Surplus}_\delta}\leq \frac{\frac{\sqrt\delta(1-\sqrt \delta)}{\alpha^2(1-\delta)}\epsilon}{\frac{\delta}{2(1-\delta)}\ln\frac{1}{\delta}}\leq \frac{2(1-\sqrt \delta)}{\alpha^2\ln(1/\delta)} ,
\end{equation}
where the first inequality comes from the fact that $\alpha\leq 1$( hence $\frac{(\alpha-1)\delta}{4\alpha^2(1-\delta)}\ln \frac{1}{\delta}\leq 0$), and second inequality is due to $\epsilon\leq \sqrt{\delta}$. Now, we have two cases:
\begin{itemize}
\item If $\alpha\geq 2\frac{1-\sqrt \delta}{1-\delta}-1$, then $\epsilon^*\leq \sqrt{\delta}$ and hence the maximum of $\Delta(\epsilon)$ happens at $\epsilon=\epsilon^*$. In this case, using inequality (\ref{ineq1}) we have
\begin{equation}
\forall \epsilon\in[0,\sqrt \delta]: \frac{\textrm{APX}_\delta(\alpha,\alpha-1+\epsilon)}{\textrm{Max-Surplus}_\delta}\leq \frac{2(1-\sqrt \delta)}{\left (2\frac{1-\sqrt \delta}{1-\delta}-1\right )^2\ln(1/\delta)}\overset{\delta\rightarrow 0} {\longrightarrow }0,
\end{equation}
\item If $\alpha< 2\frac{1-\sqrt \delta}{1-\delta}-1$, then $\epsilon^*> \sqrt \delta$ and hence the maximum of $\Delta(\epsilon)$ happens at $\epsilon=\sqrt \delta$. In this case we have
\begin{align}
\label{ineq2}
\forall \epsilon\in[0,\sqrt \delta]: ~ & \frac{\textrm{APX}_\delta(\alpha,\alpha-1+\epsilon)}{\textrm{Max-Surplus}_\delta}\leq \frac{(\alpha-1)}{2\alpha^2}+\frac{\Delta(\sqrt \delta)}{\frac{\delta}{2(1-\delta)}\ln\frac{1}{\delta}}\nonumber\\&=\frac{(\alpha-1)}{2\alpha^2}+\frac{2(1-\sqrt \delta)}{\alpha^2\ln\frac{1}{\delta}}-\frac{(\alpha+1)(1-\delta)}{2\alpha^2\ln \frac{1}{\delta}}.
\end{align}
Suppose $\alpha^*(\delta)$ be the $\alpha$ that maximizes the upperbound on revenue in (\ref{ineq2}). If $\forall C>0, \exists \delta_C$ s.t. if $\delta\leq \delta_C$ then $\alpha^*(\delta)\geq \left (\frac{C}{\ln\frac{1}{\delta}}\right )^{1/2}$, then we would have
\begin{align}
\label{ineq3}
\forall C>0, \forall \epsilon\in\left [0,\sqrt \delta\right ], \delta\in [0,\delta_C]:~&\frac{\textrm{APX}_\delta(\alpha,\alpha-1+\epsilon)}{\textrm{Max-Surplus}_\delta}\leq  
\frac{2(1-\sqrt \delta)}{\alpha^*(\delta)^2\ln\frac{1}{\delta}}\leq \frac{2}{C}.
\end{align}
As the above uppderbound holds for all $C>0$, 
so the ratio goes to zero as $\delta$ goes to zero.

Now, suppose $\exists C_0$ s.t. $\forall \delta_0$,
$\exists \delta<\delta_0$ s.t.
 $\alpha^*(\delta) < \left (\frac{C_0}{\ln\frac{1}{\delta}}\right )^{1/2}$. From (\ref{ineq2}) we have:
\begin{align}
\label{ineq4}
\forall \epsilon\in\left [0,\sqrt \delta\right ]: ~ & \frac{\textrm{APX}_\delta(\alpha,\alpha-1+\epsilon)}{\textrm{Max-Surplus}_\delta}\leq 
\frac{(\alpha^*(\delta)-1)}{2\alpha^*(\delta)^2}+\frac{2(1-\sqrt \delta)}{\alpha^*(\delta)^2\ln\frac{1}{\delta}}\nonumber\\
&\leq\frac{1}{2\alpha^*(\delta)^2}\left(\frac{4(1-
\sqrt\delta)}{\ln \frac{1}{\delta}}+\left (\frac{C_0}{\ln\frac{1}{\delta}}\right )^{1/2}-1\right).
\end{align}
We can find arbitrarily small $\delta$ such that 
 $\alpha^*(\delta) < \left (\frac{C_0}{\ln\frac{1}{\delta}}\right )^{1/2}$, 
in which case
this upperbound is a negative number. Thus, 
we know the ratio can be arbitrarily small.
\end{itemize}
\item \textit{ Case 2} ($\epsilon> \sqrt{\delta}$): In this case we have  $\xi=\epsilon$, and hence 
\begin{equation}
\textrm{APX}_\delta(\alpha, \alpha-1+\epsilon)=\frac{\delta}{(1-\delta)4\alpha^2}\gamma(\epsilon,\alpha),
\end{equation}
where $\gamma(\epsilon,\alpha)=2(\alpha-1)\ln \frac{1}{\epsilon}+4(1-\epsilon)-(\alpha+1)(1-\epsilon^2).$ We now investigate the choice of $\epsilon$ that maximizes $\gamma$ for a fixed $\alpha$.  We have
\begin{equation}
 \frac{\partial\gamma}{\partial\epsilon}=\frac{2(1-\alpha)}{\epsilon}-4+2(\alpha+1)\epsilon=\frac{2}{\epsilon}(1-\alpha-2\epsilon+(\alpha+1)\epsilon^2).
\end{equation} 
 Roots of $\frac{\partial\gamma}{\partial\epsilon}$, which are candidates for local extremum, are $\epsilon_1=1$ and $\epsilon_2=\frac{1-\alpha}{1+\alpha}$. We know if $\epsilon\geq 1$, then the mechanism cannot get any revenue as the interval $[1-\epsilon,0]$ is outside of the support of the seller's distribution. So, maximum of $\gamma(\alpha,\epsilon)$ for any fixed $\alpha$ over $\epsilon\in [\sqrt \delta,\infty]$ either happens at $\epsilon= \sqrt \delta$ or $\epsilon=\epsilon_2$. If maximum happens at $\epsilon=\sqrt \delta$ then the analysis will be the same as Case 1 and we are done. Otherwise, assume maximum happens at $\epsilon=\frac{1-\alpha}{1+\alpha}$. We have
 \begin{align}
 \label{ineq5}
&\frac{ \textrm{APX}_\delta(\alpha,\alpha-1+\epsilon)\lvert_{\epsilon=\frac{1-\alpha}{1+\alpha}}}{\textrm{Max-Surplus}_\delta}\nonumber\\
&=\frac{\frac{\delta}{(1-\delta)4\alpha^2}\gamma(\alpha,\frac{1-\alpha}{1+\alpha})}{\frac{\delta}{2(1-\delta)}\ln\frac{1}{\delta}}=\frac{1}{2\alpha^2\ln\frac{1}{\delta}}\left(  2(\alpha-1)\ln\left (\frac{1+\alpha}{1-\alpha}\right )+\frac{4\alpha}{\alpha+1}\right)\nonumber\\
&\leq \frac{1}{\alpha^2\ln\frac{1}{\delta}}\left((\alpha-1)\ln\left (1+\frac{2\alpha}{1-\alpha}\right )+2\alpha\right)\leq \frac{2\alpha+(\alpha-1)\ln(1+2\alpha)}{\alpha^2\ln\frac{1}{\delta}}.
\end{align}
\end{itemize}
Suppose $\alpha^*(\delta)$ be the $\alpha$ that maximizes the upperbound on revenue in (\ref{ineq5}) for a particular $\delta$. If $\forall C>0, \exists \delta_C$ s.t. if $\delta\in [0,\delta_C]$ then $\alpha^*(\delta)\geq \frac{C}{\ln\frac{1}{\delta}}$, using (\ref{ineq5}) we would have
\begin{align}
\frac{ \textrm{APX}_\delta(\alpha,\alpha-1+\epsilon)\lvert_{\epsilon=\frac{1-\alpha}{1+\alpha}}}{\textrm{Max-Surplus}_\delta}\leq \frac{1}{\alpha^*(\delta)\ln\frac{1}{\delta}}\leq \frac{1}{C}.
\end{align}
Since $C$ is arbitrary, we know the ratio goes to zero as $\delta$ goes to 
zero.

Now suppose $\exists C_0$ s.t. $\forall\delta_0, \exists \delta<\delta_0$ such that $\alpha^*(\delta)<\frac{C_0}{\ln\frac{1}{\delta}}$. We now have
\begin{align}
&\frac{ \textrm{APX}_\delta(\alpha,\alpha-1+\epsilon)\lvert_{\epsilon=\frac{1-\alpha}{1+\alpha}}}{\textrm{Max-Surplus}_\delta}\leq \frac{2\alpha^*(\delta)+(\alpha^*(\delta)-1)\ln(1+2\alpha^*(\delta))}{\alpha^*(\delta)^2\ln\frac{1}{\delta}}\\
\overset{\delta \rightarrow 0}{\approx}&
\frac{
2\alpha^*(\delta)+
(\alpha^*(\delta)-1)(2\alpha^*(\delta)-2\alpha^*(\delta)^2)
}{\alpha^*(\delta)^2\ln\frac{1}{\delta}}=
\frac{4}{\ln{\frac{1}{\delta}}}\overset{\delta \rightarrow 0}{\rightarrow} 0.
\end{align}
The third expression is obtained using Taylor expansion 
for $\ln(1+x)$ at $x=0$.
When $\delta$ goes to zero, 
we can always find a corresponding
$\alpha^*(\delta)$ going to zero ( $\forall\delta_0,\exists\delta<\delta_0, \alpha^*(\delta)<\frac{C_0}{\ln\frac{1}{\delta}}$).
So we may ignore the $o(\alpha^*(\delta)^2)$.\qed
\end{proof}

\begin{proof}[Proof of Theorem~\ref{theorem:priorindinapprox}]
The ``if" direction has been proved in Corollary~\ref{coro:powerdistribution}.
The proof of ``only if" direction is as follows.
First, notice that when buyer's distribution is 
uniform $[0,1]$, then 
$\phi_\mathcal{B}(v)=2v-1$,
$P(c)=\frac{c+\beta}{2\alpha}+\frac{1}{2}$.
Assume that $1-\epsilon=\alpha-\beta$. 
Since $P(c)\leq 1$, we have 
$1-\epsilon=\alpha-\beta\geq c$. First we compute $\textrm{Max-Surplus}$, and then
$\textrm{Rev-APX}$,
\begin{align*}
\textrm{Max-Surplus}=
\int_0^1 \left ( 
\int_c^1
(1-v) dv\right ) g(c) dc
=
\int_0^1 
\frac{(1-c)^2}{2}
 d G(c)
=
\int_0^1
G(c) (1-c) dc,
\end{align*}
\begin{align*}
\textrm{Rev-APX}(\alpha,\beta)&=\mathbb{E}_{v,c}\Big\{w(P(c))\mathds{1}\{v\geq P(c)\}\Big\}=\mathbb{E}_{c}\left \{\Big((1-\alpha)P(c)+\beta \Big)\Big(1-F(P(c))\Big)\right \}\\
&=
\mathbb{E}_{c}\left \{
\left ((1-\alpha)\left ( \frac{c+\beta}{2\alpha}
+\frac{\beta}{2}\right )+\beta \right )\Big(1-F(P(c))\Big)\right \}\\
&=
\mathbb{E}_{c}\left \{
\left (
\frac{(1-\alpha)(c+\alpha)+(1+\alpha)(\alpha-1+\epsilon)}{2\alpha}
\right )\left (\frac{\alpha-c-\beta}{2\alpha}\right )\right \}\\
&=
\frac{1}{4\alpha^2}
\int_0^{1-\epsilon}
\Big((1-\alpha)(c-1)+ (1+\alpha)\epsilon\Big)(1-\epsilon-c)
g(c) dc\\
&=
\frac{1}{4\alpha^2}
\Big((1-\alpha)(c-1)+ (1+\alpha)\epsilon\Big)(1-\epsilon-c)
G(c)|_0^{1-\epsilon}-\\
&~~\int_0^{1-\epsilon}
G(c) \Big ((1-\alpha)(1-\epsilon-c) -
(1-\alpha)(c-1)-(1+\alpha)\epsilon \Big )dc\\
&=
\frac{-1}{4\alpha^2}
\int_0^{1-\epsilon}
G(c) \Big((1-\alpha)(1-\epsilon-c) -
(1-\alpha)(c-1)-(1+\alpha)\epsilon)\Big)dc\\
&=
\frac{1}{2\alpha^2}
\int_0^{1-\epsilon}
G(c) \Big (
(\alpha-1)(1-c)+\epsilon
\Big )dc\nonumber\\
&=
\frac{\alpha-1}{2\alpha^2}
\int_0^{1-\epsilon}
G(c) (1-c)dc
+
\frac{\epsilon}{2\alpha^2}
\int_0^{1-\epsilon}
G(c) dc.
\end{align*}
When $\epsilon=0$, which means 
$\alpha-\beta=1$, $\textrm{Rev-APX}(\alpha,\beta)$
is $\frac{2\alpha^2}{\alpha-1}$ approximation
to Surplus. When $\alpha=2$, it gets the maximum
approximation ratio of $8$. 
If $\epsilon>0$, 
then we may consider a distribution $G$
which is supported at $(1-\epsilon,1]$, then 
we know 
$\textrm{Rev-APX}(\alpha,\beta)=0$, while
Max-Surplus is positive. So it could not approximate
Max-Surplus in this distribution.
If $\epsilon<0$, we consider two cases. 
If $\alpha\geq 1$, then 
we consider a distribution $G$ which is 
uniform on $[1+\epsilon/2,1]$, so 
we know $\int_0^{1-\epsilon/2} G(c) (1-c) dc\leq 0$, 
so $\int_0^{1-\epsilon} G(c)(1-c) dc \leq 0$, 
which means $\textrm{Rev-APX}(\alpha,\beta)\leq 0$, 
but maximum surplus is positive. 
If $\alpha<1$, then 
since $\alpha-\beta=1-\epsilon>1$, so $\beta<0$. 
Consider the distribution $G$ which is 
uniform on $[0,-\beta/2]$, then 
we know $w(P(c))=(1-\alpha)(\frac{c+\beta}{2\alpha}+\frac{\beta}{2})+\beta<0$, so the intermediary
could extract no revenue in this case. \qed
\end{proof}
\section{Proof of Lemma~\ref{iidMHR}}
\label{lemmaproofs:appendix}
\begin{proof}
Let $G(x)$ be the CDF of random variable $v=\max v_i$. We have $G(x)=(F(x))^n$ and $g(x)=nf(x)(F(x))^{(n-1)}$. Let $\tilde{h}$ be the hazard rate of $v$, hence
\begin{equation}
\tilde{h}(x)=\frac{g(x)}{1-G(x)}=n\frac{nf(x)(F(x))^{(n-1)}}{1-(F(x))^n}=n\frac{f(x)}{1-F(x)}\frac{1}{\sum_{i=0}^{n-1} (1/F(x))^i}
\end{equation}
which is non-decreasing as $h(x)=\frac{f(x)}{1-F(x)}$ is non-decreasing and $F(x)$ is non-decreasing.\qed
\end{proof}
\section{Conclusions and open questions}
In this paper we studied the problem of simple affine fee-setting mechanisms versus optimal intermediary mechanisms in the setting of 1-seller 1-buyer exchange. Our result complements the already existing result on optimality of affine fee-setting mechanisms when seller has an affine virtual cost function. In fact, we showed that under some technical assumptions, if the buyer has affine virtual value function there exist an affine fee-setting mechanism that extracts a constant approximation of optimal intermediary profit. Moreover, we showed if buyer's value is MHR and the difference between buyer's value and seller's cost is MHR, then we get constant approximation to both surplus and revenue by a constant fee-schedule mechanism.
Next, we provided inapproximability results  by showing that proper affine fee-setting mechanisms (e.g. those used in eBay and Amazon selling plans) are \emph{unable} to extract a constant fraction of optimal profit in the worst-case seller distribution. As subsidiary results we also show there exists a constant gap between maximum surplus and maximum revenue under the aforementioned conditions. Most of the mechanisms that we propose are also prior-independent with respect to the seller, which signifies the practical implications of our result.

There are many open questions left that  might be interesting for future works on this topic:
\begin{itemize}
\item Can we extend the results to the case where there are multiple sellers? 
\item As  has been conjectured in ~\cite{LN07}, affine fee-setting mechanism seem to get a good fraction of optimal revenue even under worst-case distributions of both buyer and seller. Can the proof techniques provided in this paper be used to solve that problem?
\item Can we generalize techniques provided in this paper to other exchange environments such as multi-item environments?
\end{itemize}

\end{document}